\documentclass[10pt]{statsoc}
\usepackage[a4paper]{geometry}
\usepackage{graphicx}
\usepackage[textwidth=8em,textsize=small]{todonotes}

\usepackage{xr}
\makeatletter

\newtheorem{definition}{Definition}
\newtheorem{proposition}{Proposition}
\newtheorem{remark}{Remark}
\newtheorem{lemma}{Lemma}
\newtheorem{corollary}{Corollary}
\newtheorem{theorem}{Theorem}

\usepackage{fancyhdr}

\usepackage{amsmath}
\usepackage{lmodern}

\usepackage{times}
\usepackage{bm}
\usepackage{natbib}

\usepackage[plain,noend]{algorithm2e}

\makeatletter
\renewcommand{\algocf@captiontext}[2]{#1\algocf@typo. \AlCapFnt{}#2} % text of caption
\def\@algocf@capt@plain{top}
\renewcommand{\algocf@makecaption}[2]{%
  \addtolength{\hsize}{\algomargin}%
  \sbox\@tempboxa{\algocf@captiontext{#1}{#2}}%
  \ifdim\wd\@tempboxa >\hsize%     % if caption is longer than a line
    \hskip .5\algomargin%
    \parbox[t]{\hsize}{\algocf@captiontext{#1}{#2}}% then caption is not centered
  \else%
    \global\@minipagefalse%
    \hbox to\hsize{\box\@tempboxa}% else caption is centered
  \fi%
  \addtolength{\hsize}{-\algomargin}%
}
\makeatother

\usepackage[title]{appendix}
\usepackage[margin=2cm,font=small]{caption}

\usepackage{multicol}
\usepackage{csquotes}
\MakeOuterQuote{"}\EnableQuotes

\input xy
\xyoption{all}

\usepackage{cancel}

\usepackage{enumerate}
\usepackage{lmodern}
\usepackage{graphicx}
\usepackage{multirow}
\usepackage{amssymb}
\usepackage{soul}
\usepackage{color}

\def\F{{\mathcal F}}

\newcommand{\calv}{\mathcal V}

\newcommand{\calf}{\mathcal F}
\newcommand{\calu}{\mathcal U}
\newcommand{\calb}{\mathcal B}
\newcommand{\call}{\mathcal L}
\newcommand{\caln}{\mathcal N}
\newcommand{\cale}{\mathcal E}
\newcommand{\calg}{\mathcal G}

\newcommand{\calp}{\mathcal P}
\newcommand{\calc}{\mathcal C}
\newcommand\ind{\protect\mathpalette{\protect\independenT}{\perp}}
\def\independenT#1#2{\mathrel{\rlap{$#1#2$}\mkern4mu{#1#2}}}

\title[Graphical criteria for identification of causal effects in continuous-time event-history analyses]%Causal graphs for  identification of causal   effects in continuous-time event-history analyses]
{Graphical criteria for the identification of\\ 
marginal causal   effects in continuous-time survival \\
and event-history analyses}
\author{K. R{\O}YSLAND}
\address{Department of Biostatistics, University of Oslo, Norway}
\email{kjetil.roysland@medisin.uio.no}

\author{P. RYALEN}
\address{Department of Biostatistics, University of Oslo, Norway}
\address{Department of Mathematics, EPFL, Lausanne, Switzerland} 
       
\author{M. NYG{\AA}RD}
\address{Cancer Registry of Norway}

\author{V. DIDELEZ}
\address{Leibniz Institute for Prevention Research and Epidemiology -- BIPS, and
Department of Mathematics and Computer Science, University of Bremen, Germany}

\begin{document}

\pagestyle{plain}  % Simple page numbers at the bottom of the page


\begin{abstract}
        We consider continuous-time survival and  event-history settings, where our aim is to {\color{black} graphically represent causal structures allowing us to characterise  when a causal parameter  is {\em identified} from observational data}. 
        This {\color{black} causal parameter} is formalised as the effect on an outcome event of a (possibly hypothetical) intervention on the intensity of a treatment process, i.e.\ a stochastic intervention. 
        To establish {\color{black} identifiability}, i.e. whether valid inference about the interventional situation can be drawn from  observational  data, we propose {\color{black} novel}  graphical rules indicating whether the observed information is sufficient to obtain the desired causal effect by suitable re-weighting. {\color{black} This requires a different type of graph than in discrete time}.  {\color{black} We formally define causal semantics for the} corresponding  
        dynamic graphs that represent local independence models for multivariate counting processes. 
        Importantly, {our work} highlights that causal inference from censored data relies on {subtle} structural assumptions on the censoring process beyond independent censoring. These {\color{black}  assumptions} can also be represented and verified graphically. Put together, our results {\color{black} are the first to} establish {\color{black} graphical rules for} non-parametric identifiability  {\color{black} of hypothetical interventions on event processes in this generality for the continuous-time case, not relying}  on particular {\color{black} (semi-)parametric}  survival models.
        We conclude with a data example on HPV-testing for cervical cancer screening, where {\color{black} the  assumptions are illustrated graphically and} the desired effect is estimated by re-weighted cumulative incidence curves.
\end{abstract}        

% \newpage

\bigskip
\section{Introduction}
Survival analysis is a fundamental field of biostatistics. The typical aim of many medical or epidemiologcial studies is to investigate, e.g.,  how to delay the event of death,  progression of disease or  other untoward occurrences. 
{\color{black} When the research question is about the behaviour under certain {\em changes} to the processes,}  e.g.\ due to some intervention or manipulation, 
we consider this as  {\em causal inference}, in accordance with a wide literature
\citep{rubin:74,robins:86,sgs:00,Pearl,dawid2010identifying,peters:16}. 

{\color{black} 
Adopting the framework of counting processes \citep{Andersen,AalenGjessingBorgan,royslandmsm},
 we here provide a novel formal and graphical framework for causal reasoning about event-histories in continuous time. The proposed causal notion relies on formalising the intended interventions as modified intensities of the relevant continuous-time processes.} This reflects, for example, early versus late treatment initiation, or  higher versus lower frequency, say, of radiation therapy \citep{Ryalen18}. 
 {\color{black} However, while evaluating the effect of such an intervention may be the ultimate aim, we here focus on a key requirement for any causal analysis which is} to establish whether valid inference about the interventional situation can, at least  in principle, be drawn from the available data; latent processes may pose an obstacle, even in randomised studies, especially when they induce unobserved time-dependent confounding \citep{robins:86}. Hence, non-parametric  {\em identification} of the target of inference should be ensured \citep{manski2003partial,PearlShpitser}. Formulating and checking assumption for identification has benefited hugely from graphical representations, specifically causal DAGs \citep{pearl:95:diagram,robins:01a,dawid:02}, where a complete graphical characterisation of identifiability based on $d$-separations is available \citep{PearlShpitser,ShpitserPearl:08}. 
 {\color{black} 
 While these have been extended to discrete-time settings with time-dependent treatments \citep{pearlrobins:95,dawid2010identifying}, } neither causal DAGs nor these criteria can easily be transferred to continuous-time situations modelled with stochastic processes,  exhibiting feedback and censoring \citep{AalenRoyslandRSSA}. To remedy this shortcoming, so-called local independence graphs and the  notion of $\delta$-separation have been suggested as alternative representations of (in)dependencies between processes \citep{Didelez06,DidelezRoyalSB}; so-far, these have lacked an explicit causal semantic despite being used in causal contexts \citep{roysland:AOS,mogensen2020}. 

{\color{black} The novel and central contribution of the present paper} is the concept of {\em causally valid local independence graphs}, and  {\em a graphical criterion for non-parametric identification of causal parameters}, i.e. aspects of the interventional distribution of the continuous-time event processes. 
{\color{black} We introduce a general notion of `eliminable' processes which can be marginalised without destroying the causal structure and which can be checked graphically.}  
While the ensuing conditions for {\color{black} non-parametric identification} essentially demand the absence of unobserved (time-dependent) confounding,  {\em censoring} presents an added complication as it further limits, and possibly biases, the information provided by the observable data. 
As we show, censoring must not only be independent in a probabilistic sense \citep{AndersenEncy} but  also in a causal sense (to be formalised later), i.e. it must allow inference for a situation where censoring is prevented \citep{hernanrobinsbook}. Our results on graphical criteria for identification thus encompass issues of time-dependent confounding as well as right-censoring.

{\color{black} Once it has been established  that the desired causal  parameter is identifiable, we propose to fit a  {\em continuous-time marginal structural model (MSM)}, which
can be done  non-parametrically} by a weighting method \citep{ryalen2019additive}. Here, the continuous-time weights are similar to inverse-probability of treatment (IPTW) and censoring weighting (IPCW) known from discrete-time MSMs \citep{Robins1,Zidovudine}. In fact, the latter can be considered as an approximation to the continuous-time weights. The weighting approach is based on a change of measure technique known from stochastic calculus and financial modelling \citep{royslandmsm}.

The outline of our paper is as follows. We begin with some background, focusing on local independence as a dynamic notion of independence and its graphical representation. Then, we propose our central notion of causal validity including an example how it  may fail. Proposition \ref{prop:LR}  establishes the key role of  re-weighting with a likelihood-ratio process to obtain the interventional distribution. The central result providing  graphical rules for non-parametric identifiability is given in Section \ref{sec:ident}. Proposition \ref{proposition:indCens} then gives general sufficient (graphical) conditions, essentially only assuming that intensities exist, for  causal validity in the context of censoring. Section \ref{sec:msm} joins all these pieces for  main result on how to fit a marginal structural model using suitable re-weighting privded a sufficient set of covariate processes. We conclude by an illustration of our approach with the example of HPV-testing for cervical cancer screening; this is a more advanced analysis, and provides a formal justification for the analysis in \cite{NygaardRoysland}.

\section*{Set-Up  and Notation}

We consider a collection of independent individuals, say, patients,  observed over time $t\in [0,T]$, where $T$ is end of follow-up, and $t=0$ denotes baseline. There may be baseline measurements, but events of interest occur over time: the outcome event(s), e.g.\ survival; a treatment event, e.g. start or switch of treatment; and possibly other events such as the occurrences of side-effects. A special event is censoring. All  events are represented by counting processes, where $N_t^i$ indicates how often an event  $i$ has occurred by time $t$. Where relevant, we specifically denote by $N^y$ and $N^x$ (omitting the subscript $t$) the outcome and treatment processes, respectively. The process $N^c$ denotes the counting process for censoring; sometimes we refer to $C$ as the jumping time of $N^c$. We assume throughout that the compensators of the considered counting processes are absolutely continuous,  implying that the counting processes jump at different times (with probability 1) and  form a multivariate counting process \citep{Andersen}. When referring to \textit{the} intensity of a counting process $N$ we are tacitly referring to a version of the predictable intensity of $N$. We will repeatedly exploit that these intensities, combined with the distribution over the baseline variables, uniquely determine the distribution $P$ over the considered variables and processes \cite[Theorem III 1.26]{JacodShiryaev}.

A continuous-time MSM, as introduced formally in Section \ref{sec:msm}, models the effect of a strategy for generating $N^x$ on some aspect of $N^y$. Let $P,\lambda$ denote the observational distribution and intensities, then we use $\tilde P, \tilde \lambda$ for quantities under the interventional regime  enforcing that strategy.
For example, $\tilde S(t)=\tilde P(N_t^y=0)$ might denote the survival probability under a strategy such as `early treatment initiation' specified by a given intensity $\tilde \lambda^x$ for $N^x$, where $\tilde \lambda^x$ is enforced by the intervention. In order to be able to evaluate such effects from observational, i.e.\ non-interventional, data generated from $P$,  time-varying confounding needs to be addressed. We will formulate {\color{black} probabilistic and graphical} criteria characterising which information can be ignored without introducing bias.
In order to address questions relating to unobservable variables or events, we will `move' between different sets of information: $\calv$ denotes the set of all variables and processes deemed `relevant' to a system, including the censoring process;  this includes baseline variables $X\in \calb$, as well as counting processes $N\in \caln$. Here, `relevant' refers to the requirement that the resulting model is causally valid as formalised in Section \ref{sec:causalval} below.
Throughout, we are concerned with interventions on processes, not in baseline variables; the latter are included to define subgroups of interest, or because they may be needed to adjust for baseline confounding. Note that an intervention on a process does not affect the distribution of baseline variables as a cause has to come before the effect in time.

It may not be possible to  observe the complete set of variables and processes, or we may not be interested in all of these. Thus, $\calv_0=\calb_0\cup\caln_0$ denotes the subset of baseline variables and processes that we focus on and can observe; these exclude $N_t^c$ as inference typically aims at an uncensored situation (to be formalised). We will sometimes use the notation $\dot\cup$ for the union of  pairwise disjoint sets. %; since $\calb_0$ are baseline variables and $\caln_0$ are processes we have $\calv_0=\calb_0\dot\cup\caln_0$.
Where appropriate, we explicitly refer to sets of unobserved variables or processes by $\calu$, and to sets of observed processes that are not of primary interest as $\call$. Thus, the analysis is marginal over $\calu$ and the causal parameters may further be marginal over $\call$. 
Note also that often we will not distinguish between baseline variables and processes, as a variable can be understood as a constant process on $t>0$. We will highlight where it is important to distinguish baseline variables from processes.

Throughout,  we use $\calf$ (or $\calf_t$) for a $\sigma$-algebra (or filtration) on the full set of information generated by $\calv$.
To improve readability, we will prefer the notation $\calf$ for the filtration (dropping the subscript $t$) unless time needs to be explicitly referenced.
When a $\sigma$-algebra (filtration) is generated by a subset $A \subset \calv$ of processes we write $\calf^{A}$ ($\calf_t^{A}$), where we also use $\calf^{a,b}$ instead of $\calf^{\{a,b\}}$. Here, we associate processes/variables with their indices  so that the filtration generated by $N^a$ is denoted by $\mathcal F^a$ and the filtration generated by $N^a$ and $N^b$ with $\calf^{a,b}$, etc. Finally,  $P|_{\calf^A}$ denotes the restriction of a probability measure to the reduced set of information $\calf^A$.

\section{Graphical local independence Models} 
\label{sec:LIM}

Graphical local independence models (or, local independence graphs) have been suggested and investigated by \cite{Schweder,Aalen87,didelez2,DidelezRoyalSB,mogensen2020}. They can be seen as stochastic processes' counterpart  to directed acyclic graph models or Bayesian networks  \citep{lauritzen1996}: {  local independence graphs have directed edges, but allow for cycles to represent dynamic feedback}. They graphically encode the probabilistic (in)dependence structure between processes. Unlike Bayesian networks, which represent conditional independencies between variables, the independencies between processes { (the nodes)} are in terms of local independencies. Informally this means that, at any time, the presence of one process does not depend on the past of another process given some other information on the past; hence this is an asymmetric notion of independence \citep{Didelez06}. In the graph below, for instance, we represent that the process $N^a$ is locally independent of the process $N^b$ but not vice versa:
$$
\xymatrix{
N^a\ar@/^/[r]  & N^b \\
}.
$$ 
Note that there would be a {\em second} directed edge between the two nodes, from $N^b$ to $N^a$, if the processes were mutually locally dependent on each other, i.e. if there was {\em feedback}.

In this section, we give  basic background information on local independence graphs  before we formalise the corresponding causal semantics in Section \ref{sec:causalval}.

\subsection{Intensities  and local independence} 

A key concept is that of the intensity  of a counting process $N_t$, which throughout we assume to exist for all counting processes. It establishes the dependence between the process' present, or short-term prediction, and `past information'. Here, we review the concept briefly; a precise mathematical treatment requires tools from martingale theory, see
\cite{AalenGjessingBorgan} or \cite{JacodShiryaev}.

Let $\calf$ be a filtration generated by a set of  variables and processes, e.g. those in a set $\calv$ including $N$, then the $\calf$-intensity of $N$ satisfies
\begin{equation}
	\label{intensity process} 
	    E( dN_t | \mathcal F _{t-} ) = \lambda_t \, dt.
	\end{equation}

Crucially, the intensity  depends on what past information we include. 
For instance, the $\calf$-intensity  and the $\calg$-intensity, for a reduced information $\calg_t\subset \calf_t$, of the same process $N$ are  not necessarily identical; the latter can be obtained through the Innovation Theorem \cite[II.4.2]{Andersen}. 

However, when the intensity for a reduced set of past information does remain the same, we speak of a local independence.
Specifically, consider  {\color{black} $N^a \in \caln$, a set $\calc$ of baseline variables and processes excluding $a$, and a single $b$ indexing a process or a baseline variable. Then we say that $N^a$ is {\em locally independent} of  $N^b$ given $\{a\} \cup \calc$ if the $\calf^{\{a,b\} \cup \calc}$-intensity of $N^a$ is indistinguishable from its $\calf^{\{a\} \cup \calc}$-intensity. }
%$N^a, N^b \in \caln$ {\color{black}  and let $\calv=\calc \cup \{a,b\}$. Then we say that $N^a$ is {\em locally independent} of  $N^b$ given $\calc\cup \{a\}$ if the $\calf^{\calv}$-intensity of $N^a$ is indistinguishable from its $\calf^{\calv\backslash b}$-intensity, where $b$ can also index a baseline variable. }
%$\calv\backslash\{b\}$ if the $\calf^{\calv}$-intensity of $N^a$ is indistinguishable from its $\calf^{\calv\backslash b}$-intensity (where $b$ can also index a baseline variable).
Thus, local independence formalises the intuitive notion that the short-term prediction of $N^a$ is unchanged when removing information on the past of $N^b$ as long as past information generated by variables and processes in {\color{black} $\calc$} %\calv\backslash\{b\}$ 
and its own past, is given. To make this last point explicit, we always  include the process $N^a$ itself in the conditioning set (note, \cite{DidelezRoyalSB} uses the same concept but slightly different notation).
Alternatively, we can  define local independence using martingales: $N^a$ is locally independent of $N^b$, given {\color{black} $\{a\} \cup \calc$},
%\mathcal V \backslash \{b\}$, 
if $N^a_{t} - \int_0^{t} \bar \lambda_s ds$ is a local martingale with respect to $\mathcal F^{\mathcal {V}}$ when $\bar \lambda$ is the $\mathcal F^{\color{black} \{a\} \cup \calc}$-intensity of $N^a$. %\mathcal {V} \backslash b
We write local independence as $N^b \nrightarrow N^a \,|\, N^{\color{black} (\{a\} \cup \calc)}$ or  $b \nrightarrow a \,|\, {\color{black} (\{a\} \cup \calc)}$ for short. %v\backslash \{b\}

As a convention, we always include a process' own history in the filtrations for its different intensities. A detailed treatment of local independence can be found in \cite{Aalen87} and \cite{DidelezRoyalSB}  with further generalisations by \cite{mogensen2020}. 
Some important properties of local independence are (i) that it is asymmetric and (ii) it is relative to the set of given information.  If $N^a$ is locally independent of $N^b$ given {\color{black} $\{a\} \cup \calc$} then this does not necessarily imply the converse, nor does it imply local independence for a subset {\color{black} $\calc_0 \subset \calc$}. 

\subsection{Local independence graphs and models} \label{sec:graphmod}

We now turn to the graphical representation of local independence structures. { For formal details on local independence graphs and models see \citep{Didelez06,DidelezRoyalSB}.}

A graph  $G=(\calv,\cale)$ is  given by a set of vertices (or nodes) $\calv$ and { directed} edges $\cale$; the nodes represent variables or processes; there can be up to two edges between nodes representing dynamic relations. The induced subgraph $G_A$, $A\subset \calv$,  has nodes  $\calv \cap A$ and edges $\cale \cap (A\times A)$; a subgraph $G_A'$ on $A$ is given if the edges are a subset of those of the induced subgraph. Any node $b\in \calv\backslash\{a\}$ with an edge $b\longrightarrow a$ is called a parent of $a$, while $a$ is a child of $b$; graphical ancestors or descendants are defined analogously in terms of sequences of directed edges.
In order to represent local independence structures, a graph should satisfy the following properties:
\begin{itemize}
\item
the node set, $\calv=\calb\cup \caln$, consists of two types, representing either baseline variables or processes;
\item
 all edges are directed;
 \item
 between  two nodes in $\caln$ there may be up to two edges, one in each direction; 
 \item
 between two nodes in $\calb$ there can only be up to one edge; 
\item
 there are no edges pointing from a node in $\caln$ to a node in $\calb$; 
 \item
 on the subset $\calb$ of baseline variables the graph is a directed acyclic graph (DAG). 
\end{itemize}
We call a graph with the above properties {\em local independence graph}.\\

A {\em graphical local independence model}, $(\calp,\calf,G)$, combines the above graph with a class $\calp$ of distributions by demanding that when there is no edge pointing from a given node to a given process then the corresponding local independence must hold for every $P\in \calp$. Among the baseline variables  (i.e.\ for their joint marginal distribution) we  require the usual directed Markov properties of conditional independence graphs to hold  \citep{lauritzen1996}. Formally, let $G$ be a local independence graph  satisfying the above properties. The corresponding graphical local independence model is a class of joint distributions $\calp=\mathcal P (G)$ for all possible outcomes or trajectories of the  nodes in $\mathcal V$ such that, under any $P\in \calp(G)$
\begin{itemize}
\item
 the $\calf$-intensity \eqref{intensity process} is well defined for each counting process in $\caln$;
 \item
 when $(b\rightarrow a) \notin \cale$, and $a\in \caln$, then $N^a$ is locally independent of $N^b$ given $N^{\calv\backslash \{b\}}$; 
 \item
 $P|_{\calf^\calb}$ satisfies the conditional independencies given by the directed Markov properties of the induced subgraph $G_{\calb}$.
\end{itemize}

Under regularity conditions the above definition implies that every counting process $N^a$ is locally independent of its non-parents, conditionally on its closure, defined as cl$(a)=\{a\}\cup\, $pa$(a)$ \citep{DidelezRoyalSB}. This means that the $\calf$-intensity of $N^a$ is indistinguishable from its $\calf^{\mbox{\footnotesize cl}(a)}$-intensity.
In the example (\ref{local_figure_1})  we find that $N^1$ is locally independent of $N^2$ given $N^{1,3}$, and that there are no other local independencies implied by the graph:
\begin{equation} 
\xymatrix{ N^1 \ar@/^1pc/[drr] \ar@/^/[dr] & & \\ &
		N^2 \ar@/^/[r] & N^3 \ar@/^/[l] \ar@/_2.5pc/[ull]
	} \label{local_figure_1} 
\end{equation}
A further example, combining baseline variables and processes, is given in Supplement \ref{appendix:li_examples}.

\subsection{$\delta$-separation }
\label{sec:graphsep}

In order to use graphical local independence models for causal reasoning, and especially to assess identification of causal parameters, we need to be able to read off independencies  retained, or dependencies introduced, when marginalising over possibly unobservable / latent variables or processes, e.g.\ to check if these unobservables could induce confounding bias. For instance in the above example graph (\ref{local_figure_1}), if  $N^3$ were unobservable it could induce a (marginal) local dependence of $N^1$ on  $N^2$. This type of property can be read off from a local independence graph by means of $\delta$-separation \citep{Didelez06,DidelezRoyalSB}, in analogy to $d$-separation for DAGs. Due to the asymmetric nature of local independence, $\delta$-separation must also be asymmetric and is therefore different from $d$-separation. 

Before formally defining $\delta$-separation, we require the notions of  `blocked trail' and `allowed trail'. 
A trail  is a subgraph of $G$ formed by unique vertices
        $\{v_{0}, \dots, v_{m}\} $ and edges $\{e_{1}, \dots, e_{m}\}$ such that either $e_{j} = v_{{j-1}}
\rightarrow v_{j}$ or $e_{j} = v_{{j-1}} \leftarrow v_{{j}} \in \cale$ for every $j = 1, \ldots, m$. The trail is said to start in $v_{0}$ and end in $v_{m}$. As there can be multiple edges between nodes, there can be different trails on the same set of nodes. 
A trail is said to be \emph{blocked} by a set of  vertices $C\subset \calv$ if either (i)	$C$ contains a vertex $v_{j}$, $j\in \{2,\ldots , m-1\}$, on the trail such that $e_j=	v_{j}\rightarrow v_{j+1}$ or $e_{j-1}=v_{j-1}\leftarrow v_{j}$ on the trail (i.e.\ $v_{j}$ is a non-collider), or 
	(ii) the trail contains the edges $v_{j-1}\rightarrow v_{j}\leftarrow v_{j+1}$ such that $C$ contains neither  $v_j$ nor any of its descendants. 
Otherwise, the trail is said to be \emph{open} relative to $C$.  
An {\em allowed} trail from a node $v_0$ to a node $v_m$ in $\caln$ is a trail ending with a directed edge into the node $v_m$, i.e. $e_m=v_{m-1}\rightarrow v_m$.

In DAGs, for disjoint sets  $A,B,C$, we say $A$ and $B$ are  $d$-separated by $C$, if every path between $A$ and $B$ is blocked by $C$; this separation is symmetric in $A$ and $B$. For distributions that satisfy the Markov properties of a DAG, every $d$-separation entails the corresponding conditional independence, i.e. $X_A\ind X_B\mid X_C$ \citep[Theorem 1.2.4]{lauritzen1996,Pearl}. In a graphical local independence model, this is still the case for baseline variables $A,B,C\subset \calb$, but in addition we use $\delta$-separation to read off local independencies as follows \citep{Didelez06,DidelezRoyalSB}.

\begin{definition} \label{def:deltasep}
Let $a \in \caln $ and $B, C\subset \calv$ be disjoint  subsets of vertices in a local independence graph $G$. 
		Then $B$ is $\delta$-separated from $\{a\}$ by $\{a\} \cup C$ if every 
		allowed trail from any $b\in B$ to $a$ is blocked by $C$.
		For a subset $A \subset \caln$, $B$ is $\delta$-separated from $A$ by $A\cup C$
 if every allowed trail from any $b\in B$ to any $a\in A$ is blocked by $(A\cup C)\backslash \{a\}$. We then  write $B \nrightarrow_G A \mid A\cup C$. 
	\end{definition}

The role of $\delta$-separation is in guaranteeing (under regularity conditions) a corresponding local independence in the model \citep[Theorem 1 and 3.4]{DidelezRoyalSB}:  whenever $B \nrightarrow_G A \mid A\cup C$  then the sub-process $A$ is locally independent of the processes (or variables) in $B$, given $A\cup C$. As $A\cup B\cup C$ does not need to equal $\calv$, $\delta$-separation allows us to infer marginal local independencies in subsets of the system. 
Note that the `blocking of allowed trails' condition has an equivalent `moral graph' condition which allows to check $\delta$-separation on an undirected graph \citep{lauritzen1996}; this and additional information on  $\delta$-separation can be found in \cite{Didelez06}.

\begin{remark}
The definition of local independence graph and $\delta$-separation takes as implicit that every process always depends on its own past so that we do not make use of any self-loops and {\em always condition on a process' own past}. 
\cite{mogensen2020} make such a distinction; moreover the authors generalise their treatment to dependencies due to latent processes shown graphically as bi-directed edges and self-loops. The corresponding notion of separation is called $\mu$-separation.
See Supplement \ref{appendix:li_examples} for further examples.
\end{remark}

\begin{remark}\label{faithf}
The above notion of local independence graph ensures that whenever there is a $\delta$-separation in the graph then there is a local independence in the model  $\mathcal P (G)$. If the converse holds, i.e. {\em all} local independencies that occur in every distribution under a probabilistic model $\calp$ can be read off from the graph $G$ via $\delta$-separation, then we say that the model is {\em faithful} to the graph \citep{Meek95}, and the graph corresponding to $\calp$ is then unique.
Faithfulness is especially relevant in the context of causal discovery \citep{mogensen18}.
For our following results, here, we do not require faithfulness; nevertheless, in slight abuse of terminology, we will  simply say `the local independence graph' when we mean the whole local independence model.  
\end{remark}

\section{Causal Validity of Local Independence Models}
\label{sec:causalval}

Graphical local independence models describe the probabilistic dynamic dependence structure of a multivariate counting process (allowing baseline variables). We now combine this with causal semantics. 
%; for instance, we may have the aim to evaluate whether a treatment should be initiated earlier rather than later
{\color{black} To this end, we need to be explicit} about assumptions that link the probabilistic structure with hypothetical interventions under which the data generating process is modified, reflecting e.g.\ the situation of earlier treatment initiation. We proceed in analogy to {\color{black}  the case of random variables linked by a causal DAG model with a factorisation of the distribution and corresponding causal Markov property. Such a causal DAG reflects the causal structure  by demanding that the joint distribution obeys the `manipulation theorem' or `truncation formula' which is a  modification of the factorised distribution \citep{Pearl,sgs:00,Did:chapter}. Below, in Section \ref{sec:causalVal}, we propose an analogous notion for the causal interpretation of local independence graphs in terms of a truncation of the corresponding factorisation which we recall first.}

\subsection{Local characteristics}
Let $\calp (G)$ be a local independence model on a set of counting processes and baseline variables $\calv=\caln\cup \calb$, with $n_0=|\mathcal N|$ and $n=|\calv|$.
%It is a consequence of Jacod's formula, see \cite{Jacod:Multivariate}, that there exist functionals $Z^1, \dots, Z^n$,  such that the joint density of each $P\in \calp$, restricted to the possible events that could occur before $t$, is uniquely determined and given by the product 
{\color{black} The joint distribution of $\mathcal N$ and $\mathcal B$ can be uniquely and  explicitly characterised on the so-called canonical space of a marked point process \citep{jacobsen2006point,LastBrandt1995marked}. This is the domain space of $\mathcal B$ multiplied by the space of all multivariate counting process realisations, so it is not restrictive to take this as the underlying probability space. For instance, if $\mathcal N$ has at most $m$ jumps (uniformly bounded) the distribution function factorizes as follows;
\begin{equation} \label{eq:recFact}
        P(\mathcal N \in d\phi , \mathcal B \in dx  )  = 
        \prod_{i = 1}^n  Z^i (\mu^i(\omega)) d\phi \nu(dx),
        \end{equation}
where $\omega = (\phi, x)$, $\phi$ takes values in a subset of $\mathbb R^{m n_0}$, $d\phi$ is the Lebesgue product measure, and $\mu^1, \dots, \mu^n$ are the {\em local characteristics}. The assumption that $\mathcal N$ has at most $m$ jumps can then be relaxed by taking limits. }
The factors in \eqref{eq:recFact} corresponding to counting processes  $i\in \caln$
can be computed from the $\calf$-intensities $\lambda^i$ and the previous jumps. We have that 
\begin{equation}
Z^i( \lambda^i)  := \prod_{s_i \leq T} \lambda_{ s_i}^i  \exp \left( -\int_0^{T} \lambda_{ s}^i  ds \right), \label{eq: Z}
\end{equation}
where $s_i$ denotes the jump times of the counting process $N^i$. 

The graphical structure is reflected in the fact that, as explained earlier, the above $\calf$-intensities $\lambda^i$ are indistinguishable from the  $\calf^{\text{cl}(i)}$-intensities, i.e.\ those generated by the past on the graphical parent-nodes of a process and its own past. 
Additionally, when $i\in \calb$, the local characteristics are not functions of time and are given by the conditional probabilities $P(X^i\mid X^{\text{pa}(i)})$ as in the factorised density of a Bayesian network \citep{lauritzen1996}. 

\textcolor{black}{A simple example of a local independence graph is
\begin{equation}
       G: \xymatrix{ 
        X \ar[r] & N^1 \ar[r]  & N^2
        } \label{eq: simple local independence graph}
\end{equation}
where $X$ is a baseline variable and $N^1$ and $N^2$ are counting processes. In this case we have that cl$(1) = \{ X, N^1 \}$ and cl$(2) = \{ N^1, N^2 \}$ (recall the notation from Section \ref{sec:graphmod}). %We thus have that
Thus, the local independencies $X \nrightarrow N^2 | N^{1,2}$ and  $N^2 \nrightarrow N^1 | (N^1,X)$ hold for every distribution in the local independence model $\mathcal P(G)$. %; see Section \ref{sec:graphmod}.
Letting $\lambda^1$ and $\lambda^2$ denote the $\mathcal F^{X,N^1,N^2}$-predictable intensities of $N^1$ and $N^2$ with respect to some $P \in \mathcal P(G)$, we thus have that $\lambda^1$ defines an $\mathcal F^{ \text{cl(1)}}$-intensity of $N^1$ and $\lambda^2$ defines an $\mathcal F^{\text{cl}(2)}$-intensity of $N^2$ with respect to $P$. 
% \begin{align*}
%     %&\lambda^1 \text{ defines an } \mathcal F^{X, N^1}\text{-intensity of } N^1 \\
%     &\lambda^1 \text{ defines an } \mathcal F^{ N^1}\text{-intensity of } N^1 \\
%     &\lambda^2 \text{ defines an } \mathcal F^{N^1, N^2}\text{-intensity of } N^2.
% \end{align*}
 %The model $\mathcal P$ is compatible with the graph if these local independencies hold for every $P \in \mathcal P$. 
 The density \eqref{eq:recFact} is given by the product $$Z^1(\lambda^1) \cdot Z^2(\lambda^2) \cdot Z^X(\mu^X), $$
 where $Z^i(\lambda^i)$ is as in \eqref{eq: Z} for $i \in \{1, 2 \}$, and $Z^X(\mu^X)$ is a density of $X$ with respect to the dominating measure $\nu$. %, i.e. $P(X \in dx) = Z^X(\mu^X(x))\nu(dx)$. 
}
%(We omit the local independence $N^2 \nrightarrow N^1 | X$ which is also implied by the graph, but follows from $N^2 \nrightarrow N^1$ e.g. using the innovation theorem).

\subsection{Causal validity}\label{sec:causalVal}

In this section, we formalise the notion of causal validity for graphical local independence models. 
Similar to most of the causal frameworks for random variables and causal DAGs, our definition reflects that some aspects of the system are considered invariant (or stable, or modular) under certain interventions on other parts of the system \citep{sgs:00,Pearl,dawid2010identifying,peters:16}. More specifically, we consider a hypothetical intervention on process $N^i$ replacing its $\calf^{\text{cl}(i)}$-intensity $\lambda^i$ by a different intensity $\tilde \lambda^i$ which is typically assumed to be  $\calf^{\calv_0}$-predictable, e.g.\ generated by a subset $\calv_0$, with the special case where it is  predictable just with respect to  its own $\calf^i$-history. 
%A typical situation might be that the local characteristic under an intervention on $i\in \calv$, $\tilde \mu^i$, imposes independence from any baseline variables or pasts of other processes (i.e.\ it is $\calf^i$-predictable), e.g.\ 
The latter  mimics the case of randomisation or exogeneity; for example, in the case where $N^i$ counts the times an
individual takes a medication or receives radiotherapy, a possible intervention could be to increase or decrease the frequency regardless of the individual’s history. Letting  the interventional intensity  be  $\calf^{\calv_0}$-predictable allows for dynamic interventions, such as a treatment being (dynamically)  intensified after the occurrence of a specific event such as a side-effect. 
% A typical causal question would be to compare two or more such hypothetical interventions, e.g.\ less frequent and more frequent radiotherapy, regarding their future effects on some other processes, typically a time-to-event outcome such as survival.

While the original model $\calp$ describes the system's natural behaviour without intervention, we denote with $\tilde \calp$ the model for the system under such an intervention. {\color{black} The latter may be obtained from the former by simple substitution of the local characteristic of node $i$ in \eqref{eq:recFact}
in analogy to the `truncation formula' \citep{Pearl,sgs:00}
as defined next.}

\begin{definition} \label{def:causalvalidity}
Let $\calp(G)$ be a graphical local independence model. {\color{black} Consider an intervention on node $i\in \calv$ (or on set of nodes $A \subset \calv$).
%The model is called {\em causally valid with respect to an intervention on node $i\in \calv$} (or on set of nodes $A \subset \calv$) if the 
We define the corresponding {\em intervention model} $\tilde \calp(G)$  by} replacing (all) $\mu^i$ by $\tilde \mu^i$ ($i\in A$) while the local characteristics of the remaining nodes remain the same in $\calp(G)$ and $\tilde \calp(G)$. Formally, if  the joint density of a specific $P\in \calp(G)$ is given by (\ref{eq:recFact}), then the corresponding $\tilde P$ is obtained as
\begin{equation} \label{eq:truncFact}
   \prod_{j\in V\backslash A} Z^j(\mu^j,t) 
   \prod_{i\in A} Z^i(\tilde \mu^i,t)
   . 
\end{equation}
\end{definition}

\bigskip

{\color{black} Whenever the above construction of $\tilde P$ from $P$ is judged appropriate in a given real-world context, we say that the model is  {\em causally valid with respect to the intervention on node $i\in \calv$} (or $A \subset \calv$). Ideally, this would be verified by actually carrying out the desired intervention.
In the absence of such experimental validation, subject matter considerations must be used to justify this way of linking the interventional regime $\tilde P$ to the observational $P$ in a given real-world context; causal validity will typically only be plausible if the system with its elements $\calv$ is sufficiently rich, e.g.\ in terms of specifying the relevant underlying mechanisms, as illustrated in Section \ref{examp:failvalid}. }

As the above definition  is based on interventions that replace intensities, it can be regarded as a `weak' notion of causality in contrast to  `strong' notions that are based on replacing equations in structural systems, e.g.\ in causally interpreted stochastic differential equations \citep{sokol14,mogensen18}.

A causally valid local independence model without processes, only baseline variables, reduces to a (locally) causal DAG \citep{Pearl}. Further, in the case of variables, we note that different choices for (\ref{eq:recFact}) can be probabilistically equivalent but imply different causal relations:
while both $P(X^1,X^2)=P(X^1|X^2)P(X^2)$ and $P(X^1,X^2)=P(X^2|X^1)P(X^1)$, only one of the two (if any) factorisations, corresponding to $X_1\leftarrow X_2$ or $X_1\rightarrow X_2$,  can be causally valid, i.e.\ either an intervention on $X^2$ affects $X^1$ or vice versa. For processes, however, the ordering is explicit in `time' so that the main issue is whether the past on the included processes and variables in $\calv$ contains sufficient information to warrant causal validity.

\subsection{\color{black} Example how causal validity can fail}\label{examp:failvalid}
We consider a  simplified clinical situation where a medical condition $U$ occurs with odds $\gamma$. 
This condition can lead to organ failure at some later point in time when the counting process 
$N^d$ jumps.  The patient may receive a treatment $N^a$ to prevent organ failure if $U$ has occurred,  as long as he has not experienced organ failure yet.
More formally we consider 
probability densities $P$ such that: 
$
P( U = 1) / P( U= 0) = \gamma
$, and 
  % \begin{align*}
    %Y^1_s & :=  I (N^1_{s-}  = 0   )\\ 
        $\lambda^a_s  := Y_s \alpha U$ and 
        $\lambda_s^d  := Y_s U$
define $P$-intensities for $N^a$ and $N^d$ with respect to $\mathcal F_t^{\{a,d,U\}}$, 
      where 
        $Y_s  :=  I (N^a_{s-}  = N_{s-}^d =  0   )$.
The example is compatible with the following local independence graph:
\begin{equation*} \label{total}
        \xymatrix{ 
        U \ar@/^/[dr] \ar@/^1pc/[drr]& & \\
          & N^a  \ar@/_/[r] & N^d \ar@/_/[l]
        }   
\end{equation*}

 We furthermore assume that we are able to intervene to prevent the use of this treatment, which means that we hypothetically force $\alpha = 0$. Moreover we assume that our model on all three nodes is causally valid with respect to this intervention, i.e. the odds of $U$ and  intensity of $N^d$ remain the same in this hypothetical scenario where the frequencies of events are governed  by $\tilde P$. 

Now assume that the role of $U$ has been overlooked by the analyst, who wrongly thinks that the submodel induced by $N^a$ and $N^d$, i.e. marginalised over $U$, is causally valid with respect to an intervention eliminating treatment.  To see that this is not the case, %the submodel induced by $N^a$ and $N^d$, i.e. marginalised over $U$, is not causally valid with respect to an intervention eliminating treatment, 
we have to show that, when ignoring $U$, the intensity of $N^d$  is not the same with as without this intervention. 
To see this, we first apply Bayes formula to obtain  
%An application of Bayes formula gives that 
\begin{align*} %\label{eq:simpleBayes}
    P ( U  = 1 | Y_s = 0 ) =
    \frac {\gamma } {\gamma  + \frac{P( Y_s = 0| U=0  )}{P( Y_s = 0| U=1  )}   }
    = \frac {\gamma } {\gamma + e^{ s ( \alpha + 1 )}  }.
\end{align*}
A similar argument shows that for the intervened system $\tilde P ( U  = 1 | Y_s = 0 )  = \gamma / (\gamma + e^{ s })$. With this we derive the intensity of $N^d$ in the intervened system; formally, let $\tilde \lambda^d$ be a version of the $\mathcal F_t^{\{a,d\}}$-intensity for $N^d$ w.r.t. $\tilde P$. The Innovation Theorem 
\cite[II.4.2]{Andersen} %together with \eqref{eq:simpleBayes} 
gives  
that 
\begin{align} \label{eq:exampleMIn}
    \tilde \lambda^d_s =  E_{\tilde P } [ \lambda^d_s | \mathcal F_{s-}^{ a, d}  ] = 
E_{\tilde P } [ Y_s U | \mathcal F_{s-}^{ a, d}  ] =  Y_s  \tilde P ( U  = 1 | Y_s = 0 ) = 
\frac {Y_s  \gamma  } {\gamma  + e^ s  }, ~\tilde P \text{ a.s.} 
\end{align}

Without the intervention and still marginally over $U$, %Erroneously assuming that the submodel imposed by just $N^a$ and $N^d$ is causally valid with respect to preventing treatment gives,  
using the same argument as in
\eqref{eq:exampleMIn}, 
%argument with the innovation theorem and \eqref{eq:simpleBayes}, 
we find that the corresponding intensity is
\begin{align} \label{eq:exampleMInEr}
    \bar \lambda^d_s = \frac { Y_s   \gamma } { \gamma  + e^{ s ( \alpha + 1   )  }}, ~  P \text{ a.s.} 
\end{align}
Hence, we see that \eqref{eq:exampleMInEr} is clearly in conflict with 
\eqref{eq:exampleMIn}, and the induced submodel is not causally valid unless $\alpha = 0$ in the observational scenario, i.e. under $P$. 

The example can be seen as a simple dynamic illustration of confounding: Ignoring the role of $U$ would  lead us to under-estimate the risk of organ failure under a no-treatment intervention because, observationally, only low-risk patients tend to remain untreated. The example demonstrates that causal validity is typically only plausible when the multivariate system of processes and variables considered contains sufficient information on `common causes' even if they are latent, such as the variable $U$ above. 
In the context of causal DAGs, the assumption that there are no omitted variables is known as `causal sufficiency' \citep{sgs:00,hernanrobinsbook}. As we will see in Section \ref{sec:ident}, certain variables or processes can, however, be ignored without destroying causal validity. \textcolor{black}{In Supplement \ref{appendix: example of faithfulness violation} we provide an more academic example where both faithfulness and causal validity are violated. }

\subsection{Re-weighting: Likelihood-Ratios and Positivity}
\label{sec:reweighting}
{\color{black} 
Recall that we want to infer from an observational setting properties under a hypothetical interventional regime. As shown in Proposition \ref{prop:LR}, below, re-weighting will play a particular role in this endeavour.}
The hypothetical regime described by $\tilde \mu^i$ can in principle be arbitrary and should be chosen to suit the considered, practically  relevant, intervention. However, if we want to learn about the hypothetical regime from  data obtained under an observational regime, the two cannot be  `too different' from each other. To formalise this we demand absolute continuity $\tilde P \ll P$, i.e\   for every event $H \in\mathcal F_T$ with  $\tilde P(H) > 0 $ we also have that  $P(H) > 0 $.
The following proposition shows that this is closely linked to the existence of a likelihood-ratio process which re-weights the observational distribution  $P$ into the interventional distribution  $\tilde P$.

\begin{proposition} \label{prop:LR}
Let $\calp(G)$ be a local independence model.
	Consider a hypothetical intervention on component  $N^*$ with $\calf^\calv$-intensity $\lambda^*$, and that this intervention imposes the new intensity $\tilde \lambda^*:= \rho \cdot \lambda^*$ where  	$\rho$ is a non-negative and predictable process. 
   Let 
    \begin{equation}  \label{eq:tdWeight}
    	W_t := \prod_{s \leq t} \rho_s^{\Delta N_s^*} \exp{ \Big\{  - \int_0^t ( \rho_s - 1) \lambda_s^* ds  \Big\}}. 
    \end{equation}
   
   The following statements  are equivalent:
	
	\begin{enumerate}
		\item \label{claim1} $\calp(G)$ is causally valid  with respect to an intervention on $N^*$, and for any $P\in \calp$, with corresponding interventional $\tilde P\in \tilde \calp$, we have $\tilde P \ll P$ on $\mathcal F_T$. 
		
		\item \label{claim2} We have that
		 \begin{equation}
		 	E_{\tilde P}\big(H \big) = 	E_{P}\big(W_t H \big)
		 \end{equation}
		for every $\mathcal F_t$-measurable variable $H$ and $t \leq T$.   
	\end{enumerate}
\end{proposition}

%The proof is given in Appendix \ref{appendix: proof of proposition 1}. 
The proposition shows that weighting the events before $t$ according to $W_t$ provides the probabilities in the hypothetical situation under the envisaged intervention.  The causal validity of the system ensures the simple structure of $W$ in (\ref{eq:tdWeight}): it only depends on the local characteristic of the intervened node itself.

For a causally valid system, $\tilde P \ll P$ if and only if the process $W$ is uniformly integrable on $[0,T]$. There exist several different conditions that imply uniform integrability, see \citet*[Theorem IV 4.6]{JacodShiryaev}  and \cite{royslandmsm} or \cite{Kallsen2002} for  more general results.

The conditions that are most relevant for us,  translate into upper boundaries on the predictable processes $\rho_t=\frac{\tilde \lambda_t^*}  {\lambda_t^*}$ and $|\tilde \lambda_t^* - \lambda_t^*|$. Such upper boundaries can also be seen as  generalisation of the positivity condition that is usually assumed \cite[3.3]{hernanrobinsbook}.  Moreover, the weights (\ref{eq:tdWeight}) can further be regarded as continuous-time version of the {\em stabilised inverse probability of treatment weights} for discrete-time marginal structural models \citep{Robins1,Zidovudine}. \textcolor{black}{A weaker positivity condition in a similar context, but still only for a discrete-time setting, has been considered by \citet{kennedy2019}.}

As a consequence (see also  Section \ref{sec:msm}), when 
$W$  is known or identified from the observable data, valid statistical analyses of the hypothetical  scenario may use  weighted averages or weighted regression analyses.
 It may then be desirable to impose further restrictions on how  $\tilde \lambda^i$ and $\lambda^i$ are allowed to differ so that the  weights are well-behaved for stable statistical inference  \citep{royslandmsm,roysland:AOS}.

\section{Causal Validity under Marginalisation}
\label{sec:ident}
{\color{black} 
We now turn to the question of when causal quantities are identified even though certain processes are unmeasured, i.e. when the data only carry information on a marginalised system.
}
Identifiability is at the core of many statistical problems, especially missing data, latent variables and causal inference  problems \citep{manski2003partial,PearlShpitser}.  
In brief, a parameter $\xi$, being a function of a distribution $P$ in a model $\calp$, is said to be identifiable from incomplete information $\calg$
if it can  uniquely be determined from $P|_{\mathcal G}$ for every $P$ in $\mathcal P$ (see Supplement \ref{appendix: identifiability} for  a formal treatment of identifiability).

A causal parameter $\xi$ is now a function of the interventional distribution $\tilde P\in \tilde \calp$ resulting from replacing the local characteristics of a given node, cf.\ Definition \ref{def:causalvalidity}. For example, we may be interested in some aspect of the survival curve under a specific intervention on some treatment process. A causal parameter always induces another parameter $\tilde \xi$ in the original model obtained by first mapping $P$ into $\tilde P$ using the above likelihood-ratio (\ref{eq:tdWeight}), and then applying $\xi$, i.e. $\tilde \xi (P)=\xi (\tilde P)$. Thus,  a causal parameter is  identifiable from incomplete information $\calg$ if $\tilde \xi$ is  identifiable from  $\calg$.

\subsection{Eliminability}

The following definition of eliminable processes  characterises graphically when certain subprocesses (or  baseline variables) can safely be ignored without destroying certain aspects of the causal local independence structure.
 It combines the notions of `sequential randomisation' and `sequential irrelevance' of \cite{dawid2010identifying}. 
Eliminability  is  then  used to establish our key results on identifiability.
 
 We require some notation, first. Here, as before, we are interested in intervening on a  {\em process} $N^*$ and consider its effect on a set of outcome {\em processes} $\caln_0$, both exclude baseline variables. Moreover, 
 an intervention on $N^*$ does not affect baseline variables as the intervention takes place after baseline. Hence, the marginal distribution on the  baseline variables $\calb$ is the same as the corresponding marginal under an intervention, i.e.  $P|_{\calf^\calb}=\tilde P|_{\calf^\calb}$. 
 However, we allow for baseline  variables as well as processes in the set over which we want to marginalise.

% \begin{definition} \label{def:eliminU}
% Let $G$ be a local independence graph with nodes $\calv= \calv_0 \dot \cup \, \calu$ for $\calv_0 = \caln_0 \cup \calb_0$ with $\caln_0\subset\caln$ and $\calb_0\subset\calb$; let $N^*\in \caln_0$, and  $\caln_0^{\backslash*}=\caln_0\backslash \{ N^* \}$. 
% Then we say that, in $G$, the set $\calu$ is {\em eliminable} with respect to $(N^*, \calv_0^{\backslash*})$ if it can be partitioned into a sequence of sets $\calu_1, \ldots, \calu_K$ such that, for each $k=1, \ldots, K$, either
% \begin{equation}
%  \calu_k 
%  \nrightarrow_{G} \caln_0^{\backslash*}
%   \,|\, (\calv_0, \bar \calu^{k+1})
%  \label{prop1(i)}
% \end{equation}
% or  
% \begin{equation}
% \calu_k  
% \nrightarrow_{G} N^*
% \,|\,  (\calv_0,\bar \calu^{k+1}),
%  \label{prop1(ii)}
% \end{equation}
% holds. Here $\bar \calu^{k+1}= (\calu_{k+1}, \ldots, \calu_K)$ and  $\bar \calu^{K+1}= \emptyset$.
% \end{definition}

\begin{definition} \label{def:eliminU}
Let $G$ be a local independence graph with nodes $\calv= \calv_0 \dot \cup \, \calu$ for $\calv_0 = \caln_0 \cup \calb_0$ with $\caln_0\subset\caln$ and $\calb_0\subset\calb$; let $N^*\in \caln_0$, and  $\caln_0^{\backslash*}=\caln_0\backslash \{ N^* \}$. 
Then we say that, in $G$, the set $\calu$ is {\em eliminable} with respect to $(N^*, \calv_0^{\backslash*})$ if it can be partitioned into a sequence of sets $\calu_1, \ldots, \calu_K$ such that, for each $k=1, \ldots, K$, either
\begin{equation}
 \calu_k 
 \nrightarrow_{G} (\caln_0^{\backslash*}\, {\color{black}, \bar \calu^{k+1}})
  \,|\, (\calv_0, \bar \calu^{k+1})
 \label{prop1(i)}
\end{equation}
or  
\begin{equation}
\calu_k  
\nrightarrow_{G} N^*
\,|\, ( \calv_0, \bar \calu^{k+1}),
 \label{prop1(ii)}
\end{equation}
holds. Here $\bar \calu^{k+1}= (\calu_{k+1}, \ldots, \calu_K)$ and  $\bar \calu^{K+1}= \emptyset$.
\end{definition}

{\color{black} Note, when $\bar \calu^{k+1}$ contains baseline variables, then  (\ref{prop1(i)}) has to be understood as $d$-separation between $(\calu_k \cap \calb)$ and $(\bar \calu^{k+1}  \cap \calb)$ gven $(\calv_0 \cap \calb)$ in addition to the $\delta$-separation  $\calu_k 
 \nrightarrow_{G} (\caln_0^{\backslash*} , \bar \calu^{k+1} \cap \caln)
  \,|\, (\calv_0, \bar \calu^{k+1})$.}

The theorem below shows that causal validity regarding an intervention on $N^*$ is retained when ignoring (i.e.\ marginalising over) eliminable sets $\calu$. This is related to and a considerable generalisation of the `non-informative treatment assignment' proposed by \cite{ArjasParner04} in the context of marked point processes; the property of eliminability also has some similarity to some principles of the selecting covariates to adjust for confounding in DAGs \citep{VanderweeleNew,Witte2019}. 
An immediate implication is that in case of eliminability, the likelihood-ratio in the subsystem and hence any corresponding causal parameter is identified from $\calf^{\calv_0}$ (see section \ref{sec:weights} below). 
When processes are unobservable it can be helpful to reassure ourselves that they are eliminable to ensure  identifiability from observables.

\begin{theorem} \label{theorem:omitU}
Consider a local independence model $\calp(G)$. Let the nodes be partitioned as in Definition \ref{def:eliminU}.  
Assume causal validity with respect to an intervention on the  process $N^*\in \caln_0$, replacing its $\calf^{\calv}$-intensity $\lambda^*$ by a $\calf^{\calv_0}$-intensity $\tilde \lambda^*$.

If $\calu$ is eliminable with respect to $(N^*,\calv_0^{\backslash*})$ in $G$, then the model restricted to $\calf^{\calv_0}$   (i.e.\ marginally over $\calu$) is also causally valid with respect to the same intervention. 
\end{theorem}

%The proof of  Theorem \ref{theorem:omitU} is given in Appendix \ref{appendix: proof of theorem omitU}. 
The conditions (\ref{prop1(i)}) and (\ref{prop1(ii)}) {\color{black} are sufficient for processes to be } ignored without destroying causal validity. Consider for instance the local independence graph
$$
U^1 \rightarrow N^* \; \substack{\leftarrow\\ \rightarrow} \; N^y \leftarrow U^2,
$$
assuming causal validity with respect to an intervention on $N^*$. First, by $\delta$-separation (\ref{prop1(i)}) we see that  $N^y$ is locally independent of $U^1$ given $(N^*,U^2,N^y)$. Second, with (\ref{prop1(ii)}), we find that $N^*$ is locally independent of $U^2$ given $N^y,N^*$. Hence, the submodel on $(N^*,N^y)$ is causally valid with respect to an intervention on $N^*$ as long as this intervention does not depend on $(U^1,U^2)$; the model essentially asserts that there is no confounding of $(N^*,N^y)$ regardless of whether $(U^1,U^2)$ are observed.

Further basic examples  where $U$ can be ignored, are described by the local independence graphs
$$
N^* \leftarrow U \leftarrow N^y
\quad \quad \quad \quad 
N^* \rightarrow U \rightarrow N^y.
$$
In the first case condition (\ref{prop1(i)}) holds, in the second condition (\ref{prop1(ii)})  of Theorem \ref{theorem:omitU}. Here, $U$ could also be  a sequence of such processes on directed paths. More  examples for local independence models with different structures of eliminable processes can be found in Supplement \ref{appendix:li_examples}.

\begin{remark}\label{rem:mogensen}
Our result on eliminability is related to the marginalisation considered by \cite{mogensen2020}. The authors propose an extended class of local independence graphs, and corresponding $\mu$-separation, which is closed under marginalisation. These more general graphs include bi-directed edges as a possible result of latent processes not shown as nodes in the graph. In our case, if we  consider $\calu$ as latent processes and if they satisfy the conditions of eliminability, then they do not induce any bi-directed edges with endpoints between $N^*$ and $\calv_0^{\backslash*}$ in these more general graphs. Moreover, their results can be used to obtain the `latent projection' graph  representing the local independence structure after marginalising over $\calu$ \citep[Definition 2.23 and Theorem 2.24]{mogensen2020}. 
In the above three examples these would be
$N^*\; \substack{\leftarrow\\ \rightarrow} \; N^y$, $N^* \leftarrow N^y$ and $N^* \rightarrow N^y$, respectively (albeit with bi-directed self loops).
However, in general it does not hold that the latent projection over eliminable nodes corresponds to the induced subgraph on the remaining nodes as bi-directed edges could occur between nodes within $\caln_0$. Latent projections of causal graphs have been used to identify valid adjustment sets \citep{witte:20} to which we return in section \ref{sec:msm}. We briefly comment on the projection graphs for the examples (\ref{eq: figure BB},\ref{eq: figure CC}) in Supplement \ref{appendix:li_examples}.
\end{remark}

\subsection{Identifiability and Likelihood-Ratio}
\label{sec:weights}

The above  Theorem \ref{theorem:omitU} immediately implies that together with $\tilde P \ll P$  the likelihood-ratio  
$$
\frac{d\tilde P|_{\mathcal F^{\calv_0}_t} } {d
        P|_{\mathcal F^{\calv_0}_t} }, 
$$
coincides with the weights of equation (\ref{eq:tdWeight})  with $\lambda_{ t}^*$ being the $\mathcal F^{\calv_0}$-intensity of $N^*$. 
Hence, as a key implication of Theorem \ref{theorem:omitU} we obtain that any causal parameter $\xi$ that is a function only of $\tilde {P}|_{\mathcal F^{\calv_0}}$ is identified from ${\mathcal F^{\calv_0}}$ without requiring information on $\calu$. 

Continuing the above example with graph  $U^1 \rightarrow N^* \; \substack{\leftarrow\\ \rightarrow} \; N^y \leftarrow U^2$: Assume $N^*$ is a process indicating  start of treatment and $N^y$  a disease process, e.g. indicating a cardiovascular event. Then $U^1$ might be a process that affecting the availability of the treatment but nothing else, e.g.\ a shortage in the pharmacy; $U^2$ might relate to events that affect the disease process but not the availability of, or decision to start, treatment, e.g.\ a change at the job. If we are interested in the effect of, say, early versus late  start of treatment on cardiovascular problems, we wish to ignore $U^1$ and $U^2$. Then, assuming the local independencies implied by the $\delta$-separations in the graph hold, and that $G$ is causally valid,  Theorem \ref{theorem:omitU} tells us that the reduced system $N^* \; \substack{\leftarrow\\ \rightarrow}\;  N^y$ is also causally valid and the likelihood-ratio is identified without information on $(U^1,U^2)$. Hence we can identify the desired causal effect by re-weighting according to \eqref{eq:tdWeight} using the $\calf^{N^*,N^y}$--intensity of $N^*$ ignoring $(U^1,U^2)$.

\section{Identifiability under Censoring}
\label{sec:cens}
{\color{black} 
While the above deals with unmeasured processes, in time-to-event settings information is often incomplete due to right censoring because, e.g.,  follow-up time of studies is limited. 
}
The issue of identifiability under right censoring can be seen from two subtly different angles: The approach via filtrations \citep{AalenGjessingBorgan} typically assumes there are processes and baseline variables of interest $\calv_0$ and a separate censoring process $N^c$, where $\calf_t^{\calv_0\cup N^c}$ is the filtration jointly generated by $\calv_0 \cup \{N^c\}$, $\calf_t^{\calv_0}$ is generated by $\calv_0$ alone, i.e.\ with no information on being censored, and $\calf_{t\wedge C}^{\calv_0\cup N^c}$ being generated by the observable processes, i.e.\ with everything that is censored being `invisible'. Identification is then about the possibility to use only information in $\calf_{t\wedge C}^{\calv_0\cup N^c}$ to infer quantities such as intensities defined with respect to $\calf_{t}^{\calv_0}$. However, this presupposes that it is self-evident what real-life situation $\calf_{t}^{\calv_0}$  represents, i.e.\ how there can be no censoring. If censoring simply occurs due to the end of follow-up, then one can say that $\calf_{t}^{\calv_0}$ refers to the patients' health and lives regardless of whether they are being observed in a study or not, and hence regardless of end of follow-up. This point of view  underlies the motivation of the  assumption of {\em independent} censoring \citep{Andersen,AndersenEncy}.  

However, when other types of events, such as `death from other causes', `treatment switching', or even just `drop-out' (often combined with a change in medical care) are considered as censoring events then it becomes less clear to what kind of situation the no-censoring filtration $\calf^{\calv_0}$ refers, let alone whether that filtration represents something practically meaningful. A different angle has therefore sometimes been adopted not only for survival analyses, but also in related problems such as missing data, drop-out or competing events \citep{young_causal_2020,Farewell17,hernanrobinsbook}: 
Assume  an overall model $P$ for $\calv$ including $\calv_0$ and $\{N^c\}$; this  is now modified to $\tilde P$ with $\lambda^c$ replaced by $\tilde \lambda^c\equiv 0$ while all other local characteristics remain the same. In other words, we assume that a sufficiently rich system $\calv$ can be conceived, such that the situation of interest, without censoring, can formally be described by a hypothetical intervention that sets the censoring intensity to zero within a causally valid system (cf.\ Definition \ref{def:causalvalidity}). This might imply a much more fundamental change than simply hiding or revealing information. We believe that this approach based on `preventing' censoring facilitates reasoning about the structural assumptions that allow identifiability of quantities in the uncensored situation, and would also highlight when censoring by certain types of events may not be practically meaningful.
Even if censoring is simply an inability to observe the system but does not in itself affect the system, then  violation of independent censoring can occur because censoring is {\em indirectly} informative for some hidden or ignored processes, and this can also easily be read off from local independence graphs thus alerting us to such a violation.

\subsection{Independent Censoring and Causal Validity}

In this section we link independent censoring to local independence, so that it can be read off from local independence graphs.  Further, we consider a hypothetical intervention on the system to {\em prevent} censoring; hence we give further conditions for identifiability  invoking causal validity. 

Independent censoring is often used informally or confused with stochastic independence between processes. Here, following \cite{AndersenEncy},  we formulate it in terms of local independence. 

\begin{definition}\label{def:IC}
Let $\mathcal P$ be a local independence model with sets of variables or processes $A \subset \caln, B\subset \calv$, and $N^c$ representing the counting process for censoring events.  Then censoring is said to be independent for  $A$, given $B$, if $A$ is locally independent of $N^c$ given $A\cup B$. In the special case where $A=\caln\setminus\{N^c\}$ and $B=\calb$, then we say that the whole model satisfies independent censoring.
\end{definition}

\begin{remark}\label{rem:indepcens}
As local independence can be read off from a local independence graph via $\delta$-separation, the above can be checked graphically. Let ${\mathcal P} (G)$ be a local independence model on a graph $G$ with sets of nodes $A \cup B \cup N^c \subset \calv$. Censoring is independent for  $A$, given $A\cup B$, if $N^c$ is $\delta$-separated from $A$ by $A\cup B$ in $G$. The whole model on $\calv$ satisfies independent censoring if the node $N^c$ has no children in $G$. Moreover, the submodel induced by $A\cup B\cup  \{N^c\} $ is subject to independent censoring if $N^c$ is $\delta$-separated from $A\cup B$ by $A\cup B$ alone in $G$ (or,  more explicitly, if $N^c\nrightarrow_G \caln \cap (A\cup B) \mid {A\cup B}$).
\end{remark}

In Supplement \ref{appendix: independent censoring and identifiability} (Lemma \ref{lem:newDecens}) we prove that  independent censoring ensures that the  intensities with respect to the uncensored filtration  are identifiable.
We further argue that as long as there is a non-zero probability to observe the event before censoring (e.g.\ by the end of follow-up) there exists \emph{de-censoring} maps $\zeta$ to obtain  $P$ on  $\calf_{t}^{\calv_0}$ from information restricted to $\calf_{t\wedge C}^{\calv_0 \cup N^c}$ (see Remark \ref{remark:decens} in Supplement \ref{appendix: independent censoring and identifiability}), i.e. 
\begin{align}
    \zeta(P|_{\calf_{T\wedge C}^{\calv_0 \cup N^c}}) =  P|_{\calf_{T}^{\calv_0}}. \label{eq: zeta map txt}
\end{align}
The above formulation via the de-censoring map is very general; in practice it is typical to show that a particular method of estimation is consistent for the desired parameter under independent censoring within a (semi-)parametric model. For our purposes, we choose to stay with the more general framework of identification just assuming the existence of intensities.  %explicit construction of \eqref{eq: zeta map txt} will depend on the particular (parametric) model assumed.

However, as can be seen from the example {\color{black} in Supplement \ref{appendix: example of faithfulness violation}, re-interpreting $N^A$ as censoring process, we note that the system $(N^A, N^D)$ satisfies independent censoring but that this is not sufficient to  ensure identifiability  for a system in which censoring is {\em prevented} as it would yield the wrong intensity of $N^D$.} We  need the additional causal validity with regard to an intervention on censoring, as  formalised next.

 \begin{proposition} \label{proposition:indCens}
         Let $\mathcal P(G)$ be a local independence model on $\cal V$ subject to independent censoring with bounded intensity for censoring; additionally  assume it is causally valid with
         respect to an intervention that {\em prevents} 
         censoring. 
         
If $N^c$ is $\delta$-separated from $\caln_0$ given $\calv_0$ in $G$, then the marginal model on ${\calv_0 \cup  \{N^c\}}$ satisfies independent censoring and retains causal validity with respect to the same intervention on $N^c$.  
\end{proposition}

The proof is given in Supplement \ref{appendix: proof of proposition indCens}.
Put together, the above results mean that when we have independent censoring we can recover the uncensored $\calf^{\calv_0}$-intensities from censored data, and if additionally we have causal validity with regard to an intervention that prevents censoring then these are also the intensities under the interventional distribution $\tilde P$ where censoring is prevented. Note that the results are general in that they do not rely on any particular structure for the intensities other than that they exist.

\subsection{`Randomising' Censoring }
{\color{black} Here we discuss a slight generalisation of the above that will make estimation more efficient. In fact,
Proposition \ref{proposition:indCens} remains true for an intervention on $N^c$ that imposes a different $\calf^{\calv_0\cup N^c}$-intensity for censoring (cf.\ the proof in the supplement); }`preventing' censoring is just a special case.
We will loosely refer to censoring interventions that impose an $\calf^{\calv_0 \cup N^c}$-intensity as `randomising' censoring. This may not be an intervention that is in itself of  practical relevance, but, as mentioned in section \ref{sec:reweighting}, leads to more stable weights such as (\ref{cx.weights}) in the following section. 
In other words, for re-weighting we want to consider intervention intensities that are `not too different' from the observational censoring intensity. The intervention that prevents censoring, corresponding to setting the censoring intensity to zero, may not yield efficient estimation. Formally, we posit the following.

\begin{remark} \label{rem:rondomcens}
We assume that if a model is causally valid with respect to an intervention preventing censoring, then it is also causally valid with respect to an intervention  randomising censoring.
\end{remark}

The following corollary justifies that we can essentially equate the two types of interventions on censoring, as an intervention that randomises censoring yields the same hypothetical distribution as one  that prevents censoring when restricted to $\calf^{\calv_0}$.

\begin{corollary} \label{corollary: prevent censoring}
Let $\calp(G)$ be as in Proposition \ref{proposition:indCens} and assume that the model is causally valid with respect to prevention of censoring, with $P^p$ being the model in which censoring is prevented. Let $P^r$ be a  model where we have imposed an $\calf^{\calv_0\cup N^c}$-intensity for censoring. Then, we have that 
$$ 
P^p|_{\calf_{t}^{\calv_0}} =  P^r|_{\calf_{t}^{\calv_0}} = P|_{\calf_{t}^{\calv_0}}. 
$$
In particular, every parameter on such hypothetical measures restricted to $\calf^{\calv_0}$ is invariant  with respect to the choice of censoring intervention. 
\end{corollary}
 
Together with Lemma \ref{lem:newDecens} (Supplement \ref{appendix: independent censoring and identifiability}) we have that the $\calf^{\calv_0}$-intensity of every $N \in \mathcal V_0$ is identifiable by the observable censored information {\em and retains its interpretation under a hypothetical scenario where censoring can be prevented}. In this case, we do not need to use re-weighting to identify parameters under prevention of censoring.
%that are naturally formulated in the counting process martingale framework (see Lemma \ref{lem:newDecens}). 
% move from $P$ to $\tilde P$. 
However, in the following Section \ref{sec:msm} we consider the case where independent censoring does not necessarily hold with respect to $\calv_0$ but may require further information, which can only be marginalised out after re-weighting.

\section{Marginal Structural Models and Censoring}
\label{sec:msm}

{\color{black} 
In this section we combine and further generalise the previous results. We wish to draw inference on the effect of a hypothetical intervention on a treatment or exposure process $N^x$ on one or more outcome processes $\caln_0$ under a further intervention that prevents censoring. The set $\caln_0$ could include a survival-type outcome, but can be much more general event-histories such as recurrent events and multi-state processes. Typically, causal validity will not hold for these sets alone, and adjustment is required for additional covariates, baseline or  processes, denoted $\call$. Therefore, the set of measured $\calv_0$ from the previous sections is now extended to $\calv_0 \cup \call$. Here  $\call$ is not of substantive interest in the sense that we would like the effect of the treatment intervention on $\caln_0$ marginally over $\call$. }
%show how the previous results can all be combined to deal with the case where we aim at estimating the effect of a hypothetical intervention on a treatment or exposure process $N^x$ on one or more outcome processes $\caln_0$ under a further intervention that prevents censoring. 
  %Moreover, we now allow the case that additional covariates, baseline or  processes, $\call$ are observable but are not of interest in the sense that we would like the effect of the treatment intervention on $\caln_0$ marginally over $\call$. 
This can be regarded as a continuous-time event-history analogue to the causal parameter of marginal structural models
\citep{Robins1,Zidovudine,royslandmsm}. The set $\call$ is typically needed for adjustment if it contains processes that are not eliminable, e.g. time-dependent confounding. In other words, the local independencies of Theorem \ref{theorem:omitU} and Proposition \ref{proposition:indCens} may not hold for  $\calv_0$ alone but would hold for $\calv_0 \cup \call$. As in discrete time, a continuous-time MSM can be fitted using suitable re-weighting. However, the re-weighting has two aspects: mimicking an intervention that prevents (or randomises) censoring (cf.\ Remark \ref{rem:rondomcens}), and an intervention on the treatment process. 

Let the stochastic system be described by the following sets of variables and processes
$$
\mathcal V = \mathcal V_0 \, \dot\cup  \,  \mathcal L \,  \dot\cup \, \mathcal U \,  \dot\cup \,  \{N^c\}
$$ 
where $\calv_0 = \calb_0 \cup \caln_0$, $N^x\in \caln_0$ is a treatment process, $N^c$ the censoring process, $ \mathcal N_{0}^{\backslash x}=  \mathcal N_0\backslash \{ N^x\}$ the outcome processes of interest, $\call = \calb_{\call} \cup \caln_{\call}$ are the measured baseline variables $\calb_{\call}$ and counting processes $\caln_{\call}$ we wish to marginalise out, and $\calu$ unobserved variables  or processes.
The interventions replace the $\calf^{\calv}$ intensities of this system by new intensities for $N^c$ and $N^x$. As we cannot observe $\calu$, we now give conditions such that we can instead work with 
the observable intensities,
%Moreover we define 
where $ \lambda^c$ is the $\mathcal F^{\mathcal V_0 \cup \mathcal L \cup N^c}$-intensity of $N^c$ with respect to $P$, while $ \lambda^x$ denotes the $\mathcal F^{\mathcal V_0 \cup\mathcal L \cup N^c}$-intensity of $N^x$ with respect to $P$.
In contrast, the interventions 
%correspond to replacing $\overline \lambda^c$ by an 
enforce an $\mathcal F^{\mathcal V_0 \cup N^c}$-predictable $\tilde \lambda^c := \rho^c \cdot \lambda^c$, and %$\overline \lambda^x$ by 
an $\mathcal F^{\mathcal V_0}$-predictable $\tilde \lambda^x := \rho^x \cdot \lambda^x$ as in Proposition \ref{prop:LR}. As Theorem \ref{theo:1}, below, shows, we can obtain  the hypothetical scenario from the observables  under key structural assumptions using the following combined weights 
\begin{equation}
\label{cx.weights}
W_t =                     %(\rho_C^c)^{I(t \geq C)}
{\color{black} \prod_{s \leq t} } (\rho_s^c)^{\Delta N_s^c}
                          \exp \left\{-\int_0^t ( \rho_s^c %\tilde \lambda_s^c  -
                           -1 )\lambda_s^c ds  \right\}
                          \prod_{s \leq t}
                          (\rho_s^x)^{\Delta N_s^x}
                          \exp \left\{-\int_0^t ( \rho_s^x 
                          - 1 ) {\lambda}_s^x ds  \right\}. 
\end{equation}
The first part of the weight refers to the censoring re-weighting, and is  given by $ \exp \{-\int_0^t  (\rho_s^c - 1) \lambda_s^c ds  \}$ before censoring.

\begin{theorem} \label{theo:1}
                With the above notation and set-up, consider a local independence model $\calp(G)$. 
We assume that censoring is independent with respect to $\calv$, i.e.\ ch$(N^c)=\emptyset$ in $G$.
 Furthermore,  assume causal validity with respect to an intervention that prevents censoring and  the   additional intervention on the treatment process.  Let $\tilde {\mathcal P}$ denote the resulting interventional model.
                       
Further, assume                    
                        \begin{enumerate}[(i)]
  \item \label{item: thm 1 delta separation} $N^c \nrightarrow_G (\caln_0, \caln_{\call}) | (\calv_0,\call)$, and  
  \item \label{item: thm 1 eliminable}  $\calu$ is eliminable with respect to $(N^x, \mathcal L \cup \mathcal V_{0}^{\backslash x})$;
         \end{enumerate}
finally, assume the technical conditions that         
$\rho^c$ and 
$\rho^x $ are bounded. Then we  have:\\ [1mm]
The interventional distribution $\tilde{P}$ under both hypothetical  interventions (preventing censoring and intervening on treatment) restricted to the subset $ \mathcal V_0$ is identified  from the  observable information $ {\mathcal F_{T \wedge C}^{\mathcal V_0 \cup \mathcal L \cup N^c}}$ by
\begin{align*}
\tilde P| _{  \mathcal F_{t}^{\mathcal V_0}} 
&= 
                          \zeta\bigg\{ \bigg( W_{t\wedge C} \cdot P|_{\mathcal F_{t\wedge C}^{\mathcal
                                        V_0 \cup \mathcal L \cup N^c}}   
                        \bigg)\bigg|_{\mathcal F_{t\wedge C}^{\mathcal V_0 \cup
                                        N^c}}
                        \bigg\}.
\end{align*}
 Thus, the marginal density $\tilde P|_{ \mathcal F^{\mathcal V_0 }_t }$ is given by re-weighting, marginalising over $\call$ and applying a de-censoring map as in \eqref{eq: zeta map txt} (see Remark \ref{remark:decens} in Supplement \ref{appendix: independent censoring and identifiability}). 
\end{theorem}              

\bigskip

The proof is given in Supplement \ref{appendix: proof of theorem 1}.
In words, the above theorem states conditions such that 
any marginal (over $\mathcal L$) causal parameters  are identified from the censored data ignoring $\mathcal U$, where these causal parameters quantify the effect on $\caln_0$  of  an intervention on $N^x$ while preventing censoring. In particular, with Theorem \ref{theo:1}(\ref{item: thm 1 delta separation}), we have independent censoring for $\tilde P$ restricted to $\calv_0 \cup \{ N^c \}$. 
{\color{black} 
And Theorem \ref{theo:1}(\ref{item: thm 1 eliminable}) can be seen as analogous to the sequential exchangability (or ignorability) assumption in discrete-time sequential treatment estimation. 
Due to the generality of Theorem \ref{theo:1}, we can use the above re-weighting strategy in any } standard survival analysis methods to estimate parameters under $\tilde P|_{ \mathcal F^{\mathcal V_0 } }$ (see also \cite{ryalen2019additive}).\\

Below, in Figure \ref{fig:theo1}, we give a graphical example to illustrate Theorem \ref{theo:1}.
Here, all nodes are processes and $\calv_0=\{N^y,N^x\}$, $\call=\{L\}$, $\calu=\{U^1,U^2,U^3\}$. The example represents a situation of time-dependent confounding by the process $L$: it is affected by, and affects itself the treatment process $N^x$ while also affecting the outcome process $N^y$.
We see that censoring is independent in $\calv$ as the node $N^c$ is childless. Property \eqref{item: thm 1 delta separation} can easily be seen via $\delta$-separation; note that while $N^c$ is locally dependent on $U^3$, the latter does not affect the remaining nodes, so that marginally over $U^3$ independent censoring is retained.
Property \eqref{item: thm 1 delta separation} would also  be invalid if $U^1$ or $U^2$ had directed edges pointing at $N^c$.
In $G$, the unobservable processes $(U^1,U^2,U^3)$ are eliminable in any sequence. We can verify that $(L,N^y)$ are locally independent of  $U^2$ given $(L,N^y,N^x,U^1,U^3)$, and $N^x$ is locally independent of  $U^1$ given $(N^x,L, N^y, U^3)$; and further $(L,N^y,N^x)$ are locally independent of $U^3$, so that  property \eqref{item: thm 1 eliminable} holds. Note that if we modified the example to $L$ being unobservable so that $\call=\emptyset$, $\calu=\{L,U^1,U^2,U^3\}$ then neither \eqref{item: thm 1 delta separation} nor \eqref{item: thm 1 eliminable} of Theorem \ref{theo:1} would hold.

\vspace{4mm}

{\em ((Figure \ref{fig:theo1} here))}

\vspace{4mm}

Theorem \ref{theo:1} can also be used in the following way {\color{black}  (similar to \citet{pearlrobins:95} for discrete time)}: Let us partition the nodes into
$$
\mathcal V = \mathcal V_0 \cup \mathcal Z \cup \{N^c\}.
$$ 
where $\cal Z$ is  any set of baseline variables or additional processes deemed relevant for causal validity. Then, if a subset $\call\subset {\cal Z}$ {\em exists} such that Theorem \ref{theo:1} holds with $\calu={\cal Z}\backslash \call$, then the interventional distribution $\tilde P$ is identified if $\call$ can be measured. We leave the development of algorithms that find such $\call$ for future work. \\
 
A slight generalisation of Theorem \ref{theo:1} can be obtained: We considered the intensities $\lambda^c$ and $\lambda^x$ of $N^c$ and $N^x$ that were measurable with respect to the observed information $\calf^{\calv_0 \cup \call \cup N^c}$ and $P$. We could instead accommodate different adjustment sets in the sense of allowing $\lambda^c$ and $\lambda^x$ to be intensities with respect to smaller filtrations, say, $\calf^{\calv_0^c \cup \call^c \cup N^c}$ and $\calf^{\calv_0^x \cup \call^x \cup N^c}$, respectively, with $\calv_0^c,\calv_0^x \subset \calv_0$ and $\call^c,\call^x \subset \call$: conditions \eqref{item: thm 1 delta separation} and \eqref{item: thm 1 eliminable}  still ensure the result. In particular, the likelihood ratio then still coincides with \eqref{cx.weights}.

\section{Application:  Introducing HPV-testing to follow-up low-grade cytology exams in cervical cancer screening program in Norway}
\label{sec:cervix}

Cervical cancer is an infrequent end-stage of common cellular changes, starting with minor  abnormalities and ranging through more definitely premalignant change to localised invasive and disseminated disease to death. This is an extremely complex process, but being able to detect cancer in its early stages or as precancers, accompanied by prompt appropriate treatment, are key elements of successful cancer screening programs.

Since 1995, Norwegian women 25 to 69 years of age are advised to attend cervical cancer screening every three years for cytology exam, with the objective to identify and treat those with cervical intraepithelial lesion grade 2 or 3 (CIN2+). Some of the cytology exams yield inconclusive results, and since 2005 HPV testing has been used to guide future treatment strategies. %Further background information can be found in Appendix \ref{appendix: application}.

\subsection{Which HPV-tests are suitable for secondary screening?} 
The three most common HPV tests in Norway from 2005 to 2010 were AMPLICOR HPV Test, Hybrid Capture2 High-Risk HPV DNA Test or PreTectTM HPV-Proofer referred to as Amplicor, HC2, and PreTectProofer \citep{NygaardRoysland}. When used after an inconclusive finding, PreTectProofer negative HPV-tests were more often followed later by a detection of CIN2+ than its competitors suggesting more false-negative tests for PreTecProofer \citep{Haldorsen,NygaardRoysland}. However, PreTectProofer patients were also subject to more subsequent testing \citep{NygaardRoysland}, presumably due to the manufacturer's recommendations. Thus, the apparent false-negative PreTectProofer results might have been due to the higher rate of subsequent testing.

The  objective of our analysis is to compare the %rate of
cumulative incidences of CIN2+ detection in the   PreTectProofer group with the  other two groups under a hypothetical scenario where an intervention ensures that the PreTectProofer patients are subject to the same rate of subsequent testing as under the   other test-types.  More formally,
%our target of inference is %a contrast of CIN2+ detection under either the Proofer HPV-test type versus the Amplicor/HC2 group, while  imposing the marginal ``subsequent test'' intensity of the Amplicor/HC2 group on all the subjects regardless of the HPV-test type they received. 
let $\tilde P$ denote the distribution under the modified ``subsequent testing'' intensity and prevention of censoring. The contrast of interest is then given as the difference in cumulative incidence functions over time $t$,
\begin{equation}
\tilde P(N_t^y = 1\mid \mbox{PTP, neg.HPV, inconcl.Cyt}) - 
\tilde P(N_t^y = 1\mid \mbox{A/HC, neg.HPV, inconcl.Cyt}).  
\label{eq:estimand}
\end{equation}
Under the assumed causal structure described next, this is only not zero if the different HPV-test types have different false-negative rates.

\subsection{Local independencies and causal validity}
Figure \ref{fig:HPVmodel} contains a local independence model for the assumed HPV-testing scenario. 
Here we have that ``latent disease'' represents a baseline variable  describing the disease at the time of the HPV-test and cytology. ``HPV-test type'' is the HPV-test that accompanies the cytology (binary: PreTectProofer, yes or no). The node ``latent progression'' represents a set of counting processes describing the disease's progression from the time of the cytology over time. ``Subsequent test'' is a counting process that `counts'  the first subsequent cytology or HPV-test. The node ``CIN2+ histology'' represents the counting process that jumps when CIN2+ is detected. The analysis is restricted to subjects who had an inconclusive (i.e.\ ASCUS/LSIL/unsatisfactory) secondary cytology screening and who initially had a negative HPV-test result, as indicated by the boxed nodes in Figure \ref{fig:HPVmodel}. 
Individuals are censored at the end of the follow-up period if there was no occurrence of CIN2+ detection. The number of deaths was negligible, and is ignored. 

\vspace{4mm}

{\em ((Figure \ref{fig:HPVmodel} here))}

\vspace{4mm}

Key assumptions are that any testing in itself does not affect the disease progression, but also vice versa, disease progression does not affect the testing regime. This could be violated if certain (undocumented) symptoms lead to the initiation of an HPV-test, but this is unlikely in the present case. %Further discussion is provided in Appendix \ref{appendix: application}.
The edge ``HPV-test type'' $\rightarrow$ ``Subsequent test'' is due to the observationally differing subsequent testing rates. In this analysis, we mimic an intervention imposing the same subsequent testing regime for all the HPV tests, in effect making "Subsequent test" locally independent of "HPV-test type" in the hypothetical scenario.  Under $\tilde P$, any association between type of HPV-test and CIN2+ histology when conditioning on $\{$``HPV-result''=negative, ``Cytology''=inconclusive$\}$ is then due to unblocked paths ``HPV-test type'' $\rightarrow$ ``HPV-result'' $\leftarrow$``Latent disease'' $\rightarrow$ $\cdots$``CIN2+ histology'' which would indicate a tendency to false-negative results due to the edge ``HPV-test type'' $\rightarrow$ ``HPV-result''.

We will appeal to Theorem \ref{theo:1}. For this we define the sets of nodes
\begin{align*}
    \calb_0 &= \{\text{Cytology}, \text{HPV-test type},  
    \text{HPV result}\} &  \\
     \calu &= \{\text{Latent disease}, \text{Latent progression} 
     \}. &
\end{align*}
Let $N^x, N^y, N^c$ be counting processes where $N^x$ counts initiation of ``Subsequent test'', $N^y$ counts  histology finding CIN2+, and $N^c$ counts censoring events. Thus, $\caln_0 = \{ N^x , N^y \}$, and $\caln_0^{\setminus x} = N^y$ and $\calv_0 = \calb_0 \cup \caln_0$. We make the following observations:
\begin{itemize}
    \item $N^c$ has no descendants, i.e. censoring is independent with respect to $\calv$.
    \item There are no allowed trails from $N^c$ to $\calv_0$, so $N^c \nrightarrow_G \caln_0| \calv_0$, and condition \eqref{item: thm 1 delta separation} from Theorem \ref{theo:1} holds.
    \item Every allowed trail from $\mathcal U$ to $N^x$ is blocked by either ``Cytology'' or ``HPV-test type'', both of which are in $\calb_0$. Hence, $\calu$ is eliminable with respect to $(N^x , \calv_0\backslash\{N^x\})$ in $G$, and condition \eqref{item: thm 1 eliminable} of Theorem \ref{theo:1} holds.
\end{itemize}
These points justify the use of Theorem \ref{theo:1}, and the $\calf^{\calv_0 }$-intensity of $N^y$ under $\tilde P$ is identified. %, when censoring is prevented, and the HPV tests are subject to equal subsequent testing regimes. 
As $\call = \emptyset$, we have from Proposition \ref{proposition:indCens} that the model restricted to $\calf^{\calv_0 \cup N^c}$ is causally valid with respect to prevention of censoring and subject to independent censoring. The censoring weights thus equal one, and \eqref{cx.weights} reduces to $W_t = \prod_{s\leq t} (\rho_s^x)^{\Delta N^x_s} \exp \big\{ - \int_0^t (\rho_s^x-1)\lambda_s^x ds \big\}$ in this example.

\subsection{Analysis} \label{section: analysis}
We consider data from the Cancer Registry of Norway on  1736 subjects (878 in the PreTectProofer group and 858 in the Amplicor/HC2 group) with inconclusive cytology and negative initial HPV test recorded in 2005-2010 until CIN2+ or end of 2010 (for details see Supplement Section 9). 
We calculate the probability of having CIN2+ detected by time $t$ in a situation where individuals receive ``subsequent test'' with intensity $\tilde \lambda^{x}$ equal to the intensity in the (pooled) Amplicor/HC2 group. The probability of interest is calculated by one minus the weighted Kaplan-Meier estimator $\hat S^w$, given by
\begin{align}
    \hat S^w_t = \prod_{T_i \leq t}\bigg(1 - \frac{\widehat W_{T_i-}^i Y_{T_i}^i}{\sum_{j} \widehat W_{T_i-}^j Y_{T_i}^j}\bigg),  \label{eq: weighted Kaplan-Meier}
\end{align}
where $\widehat W^i$ are estimates of \eqref{cx.weights}, the $T_i$'s the observed detection times of CIN2+, and the $Y_t^i$'s are at-risk indicators (not censored) at time $t$ in a given group.

For the Amplicor/HC2 group, the ``subsequent test'' intensity is equal to the observational intensity, and each $W^i$ is equal to one. Thus, \eqref{eq: weighted Kaplan-Meier} reduces to the standard Kaplan-Meier estimator with CIN2+ occurrence as the endpoint.

To estimate this probability in the PreTectProofer-group, we first need estimates $\widehat W^i$. The $\calf^{\calv_0}$-intensity of $N^x$ is only a function of ``HPV-test type'' (due to local independences implied by the graph).  We thus obtain the estimator
$$ \widehat W^i_t = 1 + \int_0^t \widehat W_{s-}^i (\theta_{s-} - 1) dN_s^{x,i} - \int_0^t \widehat W_{s-}^i I(N_{s-}^{x,i}=0) d( \hat{ \tilde A}_s^x  - \hat{ A}_s^x  ), $$
where $\theta_s = \frac{ \hat{ \tilde A}_s^x  - \hat{ \tilde A}_{s-b}^x  }{ \hat{ A}_s^x  - \hat{ A}_{s-b}^x}$ for a smoothing parameter $b$, and $\hat{ \tilde A}_s^x $ and $\hat{ A}_s^x $ are the Nelson-Aalen estimators for subsequent test initiation applied to the Amplicor/HC2 group and the PreTectProofer group, respectively. We calculate these weights using the \texttt{R} package \texttt{ahw}. 
%It is a consequence of 
The weighted Kaplan-Meier estimator \eqref{eq: weighted Kaplan-Meier} is consistent for a bandwidth parameter $b=b(n)$ depending on the sample size $n$ with $b(n) \underset{n \rightarrow \infty}{\rightarrow} 0$ and $\sup_{s \leq T}\frac{\tilde \alpha_s^x}{\alpha_s^x} < \infty$; see \citep[respectively Theorems 1 \& 2, and Theorem 1]{ryalen2019additive,ryalen2018transforming}.

\subsection{Results}
The cumulative incidences of CIN2+ detection are shown in Figure \ref{fig:triageincidence}. In the upper panel, we see three curves: for PreTectProofer without and with re-weighting, and the non-PreTectProofer group without weighting.
Thus, the proportion of CIN2+ detected in the PreTectProofer group is somewhat lower under the hypothetical subsequent-testing regime than observationally, i.e. without an intervention. However, the estimated contrast  \eqref{eq:estimand}, shown in the lower panel, is still significant. Under the structural assumptions we made, this gives support to the interpretation that there are genuinely more false-negative HPV-test results in the PreTectProofer group than with the other test types, and that the greater CIN2+ numbers are not only a consequence of more frequent subsequent testing.

\subsection{Code}
\label{subsection: code}
The code to simulate  data comparable to the real data (which cannot be disclosed) and to replicate all numerical results can be accessed at the GitHub repository\\
\texttt{https://github.com/palryalen/paper-code/tree/master}.

\vspace{4mm}

{\em ((Figure \ref{fig:triageincidence} here))}

\section{Conclusions and Discussion}

We proposed a formal graphical approach to causal reasoning in survival and general event-history settings in continuous time
{\color{black} with graphical rules for the identifiability of (marginal) causal parameters.} 
%As we have seen, this requires a dynamic notion of independence, i.e.\ local independence, corresponding generalised graphs and $\delta$-separation. The key principle of causal diagrams, i.e.\ invariance (or stability or modularity \citep{peters:16, dawid2010identifying,Pearl}) with respect to a hypothetical intervention, was extended to these local independence graphs, a property which we termed causal validity. In contrast to the common approach in causal inference, where interventions fix a variable at a single value,  we consider instead interventions that modify the intensity of a treatment process thus ensuring well-behaved weights according to Proposition \ref{prop:LR}. 
We formalised these in terms of interventions that modify the intensity of a treatment process which is similar to the notion of `randomised plans' \citep{gill2001causal} or stochastic interventions \citep{dawid2010identifying,Diaz2018}; and it is more explicit than that of `causal influence'   \citep{commenges09}. {\color{black} Further, we conjecture that our change of treatment intensity could  be thought of as a stochastic change of time: one can construct a stochastic time-change $\varrho$ such that the $P$-intensity of $N^x_\varrho$ coincides with the $\tilde P$-intensity of $N^x$ \cite[II 5.2.2]{Andersen}}. While the predominant causal approach uses potential outcomes, we have chosen to simply compare the observational and the interventional distributions, $P$ and $\tilde P$; this is similar in spirit to other causal frameworks, for instance by \cite{sgs:00,dawid2010identifying,peters:16}. 

%We derived graphical criteria to characterise which types of processes are eliminable while retaining identifiability of certain causal relations. Vice versa, this means that causal effects of the treatment process on the outcome processes are identified as long as the non-eliminable processes are observable. The criterion is thus sufficient for 
Our criteria for identifiability are sufficient, and we believe that necessary conditions analogous to \cite{ShpitserPearl:08} will be difficult to derive for the general continuous-time case. In future work it will be interesting to generalise our results to the extended local independence graphs of \cite{mogensen2020} which are closed under marginalisation and which can be obtained by projecting over unobservable processes. { 
The requirement that processes may not jump simultaneously is plausible as long as the components truly represent separate phenomena, where accidental violations due to rounding are not essential. Alternatively, one can introduce additional processes which count simultaneous events, which will work in practice as long as the treatment  and outcome processes do not jump systematically at the same time;  re-defining some of the other processes will not change the interpretation of the causal parameters of interest. If treatment or outcome processes are affected,  future work would need to explicitly integrate systematic violations, which will require a more complex likelihood ratio process. Extending  graph separation and  eliminability to this situation will also require notable technical efforts outside the scope of our paper. % and key existing works \cite{didelez2,didelez3,DidelezRoyalSB}. 
Extensions to more general processes would also be desirable; some such extensions of the graphical representation exist for larger classes of stochastic differential equations \citep{mogensen18,mogensen2022} and stochastic kinetic models \citep{bowsher2010}. Other future generalisations might allow  for processes that can jump with continuous magnitude. Carrying out such extensions requires consideration of more general counting measures than we study here. %, and in the works we rely upon \citep{didelez2,didelez3,DidelezRoyalSB}.
}

We further 
%formalised, and characterised graphically, the notion of
addressed independent censoring arguing that inference for the uncensored case additionally requires causal validity with regard to an intervention that prevents censoring. While this has been recognised 
%in the causal inference literature 
for longitudinal settings with drop-out \citep{hernanrobinsbook}, it seems less appreciated in more traditional approaches to survival analysis; an exception is recent work by \cite{rytgaard2021}.
%on continuous time-settings making explicit that the target estimand involves an intervention to prevent censoring. 
{ Our results on identifiability under censoring are more widely applicable even outside a causal inference context, for instance to determine a sufficient set of covariates for inverse probability of censoring weighting.}

Finally,  Theorem \ref{theo:1}  enables identification of causal parameters via re-weighting, %to account for time-dependent confounding and  dependent censoring in continuous-time, 
thus generalising well-known results from discrete-time marginal structural models \citep{Robins1,Zidovudine}. 
Marginal structural models are often linked to the problem of time-dependent confounding but can also be used in other scenarios \citep{joffe2004model}. In the application in Section \ref{sec:cervix}, Theorem \ref{theo:1} was used to establish that time-dependent  confounding was not an issue  (as $\call$ was the empty set) and that a fairly simple weighting process sufficed.  Similarly, \cite{Ryalen18} used a continuous-time MSM to compare the treatment regimens radiotherapy and radical prostatectomy.

Our results apply to general multivariate counting processes, which include, e.g., multi-state processes. 
%{ The requirement that processes may not jump simultaneously is plausible as long as the components truly represent separate phenomena, where accidental violations due to rounding are not essential. We further conjecture that our results can be generalised to a larger class of processes  \citep{mogensen18,mogensen2022}. However, full generality may require allowing for instantaneous dependencies as in the discrete-time case \citep{DidelezEichler}, which will make rules for weighting much more complicated.}
In particular, they do not rely on any particular (semi)-parametric class of models. While most practical inference needs additional modelling assumptions, the data example of section \ref{sec:cervix} allowed for non-parametric estimation. In addressing identifiability, we have chosen the re-weighting route which appears natural in view of the simplicity of Proposition \ref{prop:LR} and corresponds to a change of measure technique. In discrete-time settings,  g-computation is an alternative, {\color{black} or doubly-robust and machine-learning extensions thereof \citep{kallus22,luckett20,Nie20,zhang13}.} However,  g-computation  seems hard in entirely general continuous-time settings, as discussed by \cite{gill01} (see also \cite{gill2001causal}), but fully parametric versions exist \citep{gran2015}.
We believe that our graphical causal reasoning can also be combined with g-estimation  \citep{lok04,Lok}, or targeted minimum-loss estimation \citep{rytgaard2021} in continuous-time settings. It complements these methods because the graphical representation and explicit discussion of eliminability strengthens the plausibility of  {\color{black} assumptions, such as sequential (conditional) exchangeability}.

\bibliographystyle{rss}
\bibliography{mybib}

% \end{appendices}

\newpage

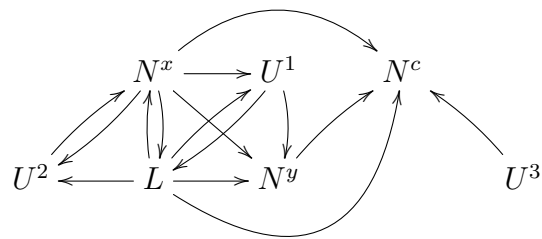
\begin{figure}[!hb]
\caption{Illustration of Theorem \ref{theo:1}}
\label{fig:theo1}

\vspace{1cm}
\begin{equation*} % Figure
        \xymatrix{
        & N^x \ar@/^.3pc/[dl]\ar[r] \ar[dr]  \ar@/^.3pc/[d]
        \ar@/^2pc/[rr] & U^1 \ar@/^.3pc/[d] \ar@/^.3pc/[dl] & N^c & \\
        U^2 \ar@/^.3pc/[ur] & L\ar[l]\ar[r]  \ar@/^.2pc/[u]
        \ar@/_3.5pc/@{->}[rru]\ar@/^.3pc/[ur]&
        N^y \ar@/^.3pc/[ur]& &
        U^3\ar@/_.3pc/@{->}[ul] 
        }
\end{equation*}
\end{figure}

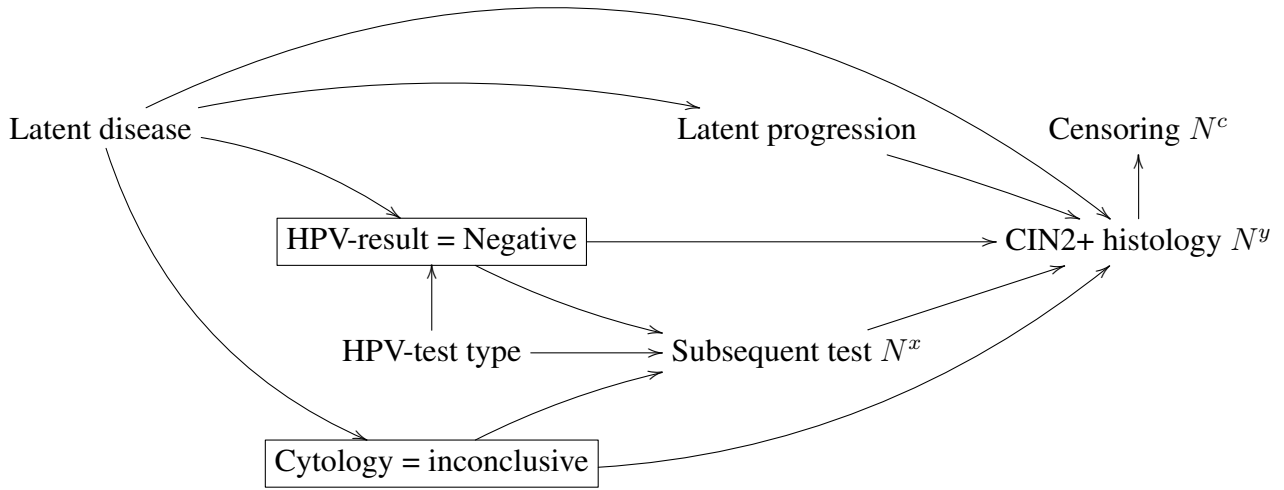
\begin{figure} 
\caption{
Local independence graph showing the assumed structure of the HPV-testing scenario. }
\label{fig:HPVmodel}
\vspace{1cm}
\begin{equation*}
\xymatrix{
\text{Latent disease} \ar@/^1pc/[dr] \ar@/_2pc/[dddr] \ar@/^5.5pc/[drrr]
\ar@/^1.5pc/[rr]   &  & \text{Latent progression} \ar@/^0.2pc/[dr]
  & \text{Censoring $N^c$} 
 & \\
& *+[F]{\text{HPV-result = Negative}} \ar[rr] \ar@/_.7pc/[dr]&  &
\text{CIN2+ histology $N^y$} %\ar@{-->}@/^0.6pc/[l] 
\ar[u] \\
& \text{HPV-test type} \ar[r] \ar[u]& \text{Subsequent test $N^x$} \ar%@/_0.6pc/ 
[ur]&
 \\
&    *+[F]{\text{Cytology = inconclusive}}   \ar@/_2.5pc/[urur] \ar@/^.7pc/[ur]& & \\
}
\end{equation*}
\end{figure}

\begin{figure}[!htp]
\centering
\caption{Upper: Proportion CIN2+ detected after the secondary screening with ASCUS/LSIL/unsatisfactory cytology and negative HPV-test. Lower: Difference between the proportions of detected CIN2+ in the PreTectProofer-group and Amplicor/HC2-group when imposing the Amplicor/HC2 group's subsequent testing regime on both groups. We obtained 95\% pointwise confidence intervals using a bootstrap sample of 400.}
\includegraphics[scale=.6]{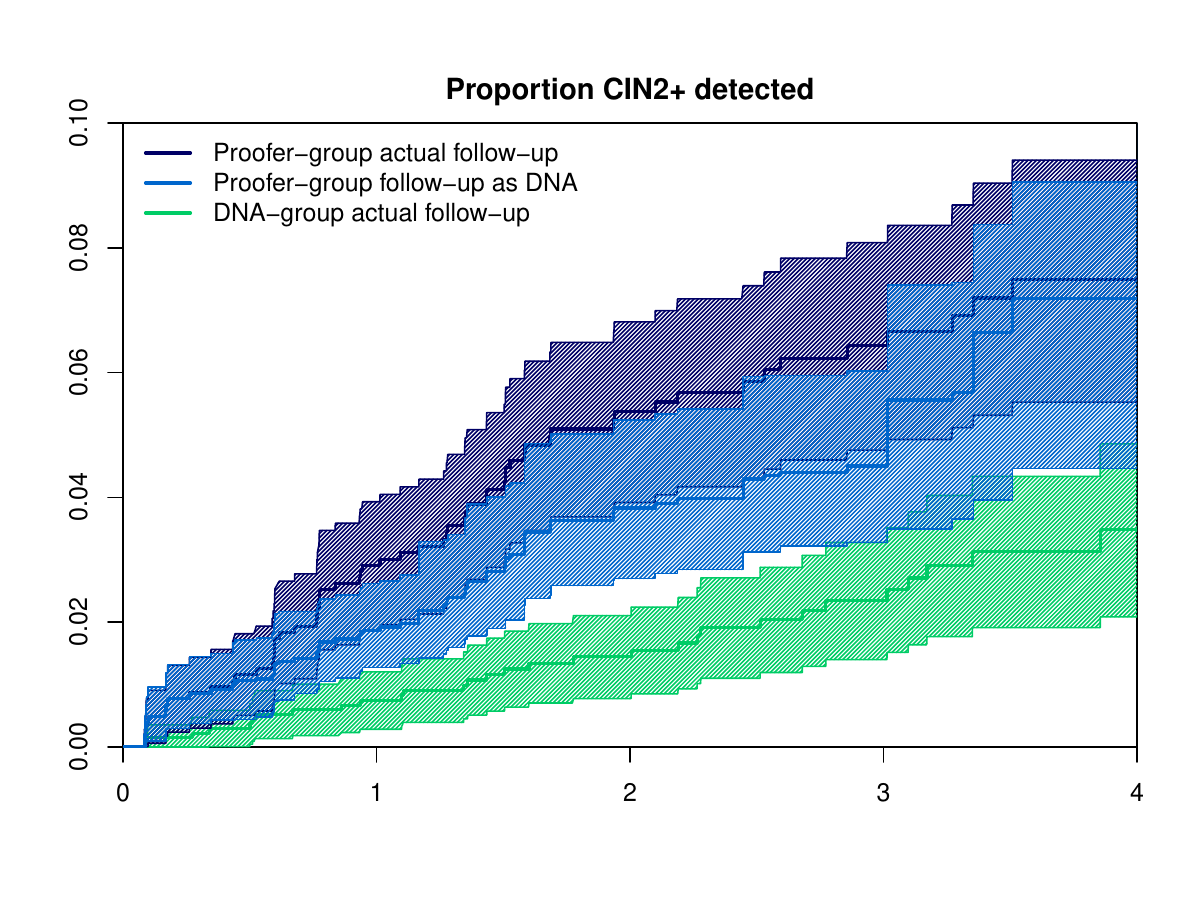}
\includegraphics[scale=.6]{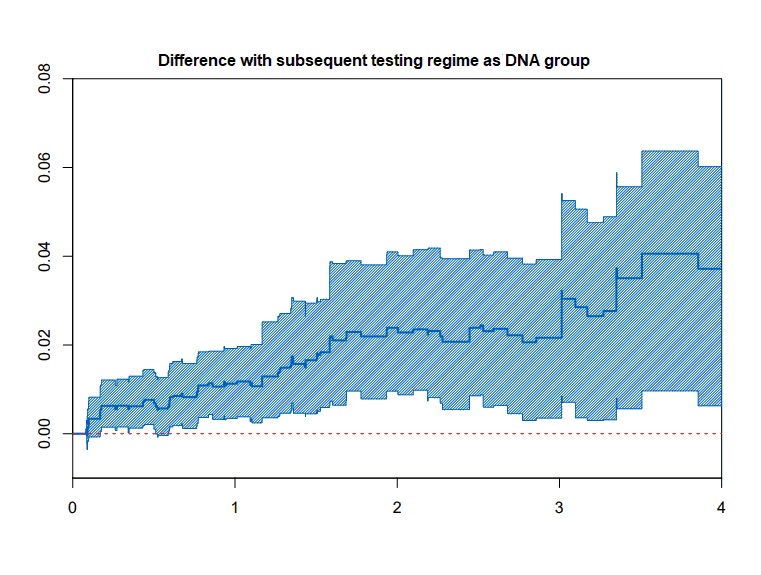}
\label{fig:triageincidence}
\end{figure}

\appendix

\section{Examples: Local Independence Graphs}
\label{appendix:li_examples}

In the graph \eqref{eq: figure AA}, below, we consider two baseline variables $X,Z$ and two counting processes $N^a, N^b$. Due to the edge $X\leftarrow Z$ there is no independence between the baseline variables. Both variables are locally independent of  $N^a$ and / or $N^b$ jointly or individually by definition. The absence of an edge from $N^a$ to $N^b$ implies the local independence  $N^a\nrightarrow N^b\mid (X,Z,N^b)$. However, as can be read of with $\delta$-separation, it does not generally hold in this local independence model that $N^a \nrightarrow N^b | N^b $ nor that $N^a \nrightarrow N^b \mid (N^b,Z)$ but it does hold that $Z$ can be ignored to obtain $N^a \nrightarrow N^b \mid (X,N^b)$. %It also holds that $N^a \nrightarrow N^b | (N^b,X,Z) $, 
Thus if $N^a$ were a censoring process, then the model as a whole and the submodel $(N^a,N^b,X)$ would satisfy independent censoring, but ignoring $X$ would violate  independent censoring, cf. Definition \ref{def:IC}.\\

\begin{align}
\xymatrix@=3em{
Z %\ar@/^/[r] 
\ar[d] \ar[r]  & N^a \\
X \ar[r] \ar[ur] & N^b \ar[u]
} \label{eq: figure AA}
\end{align}

The following are examples of local independence graphs where the processes in $\calu$ are eliminable.
In \eqref{eq: figure BB} we see that the ordering $U^{1,2},U^3$ satisfies Definition \ref{def:eliminU}, first property \eqref{prop1(i)} and then property \eqref{prop1(ii)}; by symmetry $U^{1,3},U^2$ also works (but $U^{2},U^3, U^1$, for instance, would not be a suitable ordering). The projection graph (see Remark \ref{rem:mogensen}) would just consist of $N^*$ and $N^y$  with two edges linking them.

\begin{align}
\xymatrix@=3em{
U^2 
\ar@/^/[d] \ar[r]  & U^1 & U^3 \ar[l] \ar@/^/[d] \\
N^* \ar@/^/[rr] \ar@/^/[u] \ar[ur] & & N^y \ar@/^/[ll] \ar[ul] \ar@/^/[u]
} \label{eq: figure BB}
\end{align}

\ \\ The graph in \eqref{eq: figure CC} is an example with a bivariate outcome process $N^y=(N^1,N^2)$; here the unobserved process is eliminable as \eqref{prop1(ii)} applies. The projection graph would contain a bidirected edge between $N^1$ and $N^2$ but this is not relevant for our purposes, here.

\begin{align}
\xymatrix@=3em{
N^2 
\ar[d] \ar@/^/[r] \ar@<0.5ex>[dr]  & U \ar@/^/[d] \ar[l] \\
N^* \ar@/^/[u] \ar[ur] \ar[r] & N^1 \ar@/^/[l] \ar[u] \ar@<0.5ex>[ul]
} \label{eq: figure CC}
\end{align}

\section{Identifiability}
\label{appendix: identifiability}

We need to formalise  what we mean by identifiability in our context before we present the proof of Theorem \ref{theo:1}. Our definition of identifiability given below is in the spirit of previous definitions (\cite{manski2003partial,PearlShpitser}), but we re-state it within our context and notation as it is central to our results. 
We start by clarifying the very general notion of `parameter' that we employ.

Consider a collection of probability measures $\mathcal P$ on a
$\sigma$-algebra $\mathcal F$.  We will consider computable quantities that can
be represented by maps 
\begin{equation} 
\label{eq:xiDef}\xi: \mathcal P \rightarrow \Xi.
\end{equation} 
Even if this model is not parameterised, since
we do not rely on maps $\Theta : \mathbb R^n \rightarrow \mathcal P$, we will
somewhat ambiguously refer to maps like \eqref{eq:xiDef} as \emph{parameters}.
Note that we also consider parameters that are random. For instance, 
likelihood-ratios and predictable intensity processes can thus be seen as parameters. 

Suppose we are not able to acquire all the potential information $\mathcal F$
about the whole stochastic system, but only some abbreviated adaptation
$\mathcal G$  represented by a sub $\sigma$-algebra.  We would now like to know if it is
possible to learn the value of the parameter $\xi$, based only on the incomplete information $\mathcal G$.

The probability distributions for the observable information, according to the
model $\mathcal P$, is given by restricting every $P \in \mathcal P$ to $\mathcal G$.
Whenever $P$ is a distribution in $\mathcal P$, let $ P|_{\mathcal G}$ denote
its restriction to $\mathcal G$ and let $\mathcal P|_{\mathcal G}$ denote the
set of all such restrictions.  
The usual notion of identifiability now translates into  the following.

\begin{definition} \label{identDef}
A parameter $\xi : \mathcal P \rightarrow \Xi$ is identifiable with respect to
the 
incomplete information $\mathcal G$ if $\xi(P)$ can be uniquely determined from
$P|_{\mathcal G}$, i.e. there exists another map $\eta : \mathcal P|_{\mathcal G} \rightarrow \Xi$ such that 
\begin{equation}
	\xi(P) = \eta( P|_{\mathcal G} ), \label{eq: de-censoring map}
\end{equation}
for every $P$ in $\mathcal P$. 
\end{definition}

With regard to a causal parameter, we think of the intervention changing the local
characteristics as a transformation $P \mapsto \tilde P$, and 
        $\tilde {\mathcal P} := \{\tilde P| P \in \mathcal P\}$
defines a non-parametrised model for the hypothetically intervened scenario. 
Recall also that we consider transformations where $\tilde P$ can be
achieved through re-weighting $P$, i.e. $\tilde P = \frac{d\tilde P}{dP} \cdot
P$. If the underlying model is causal, then we will refer to any parameter 
$\xi: \tilde {\mathcal P} \rightarrow \Xi$ as a \emph{causal parameter}.  
Such a parameter always induces another parameter of the original model, namely  $\tilde \xi:
\mathcal P \rightarrow \Xi $, where $\tilde \xi (P) = \xi(\tilde P)$. The
induced parameter $\tilde \xi$ is obtained by first performing the transformation
$P \mapsto \tilde P$ and then applying $\xi$ to the transformed distribution $\tilde P$. 
Pearl's notion of identifiability
\cite[Definition 3.2.4]{Pearl} now translates into the following:
\begin{definition} \label{CounterfcatIdent}
A causal  parameter $\xi:\tilde {\mathcal P} \rightarrow \Xi$
is identifiable with respect to the incomplete information $\mathcal G$ if the induced
parameter $\tilde \xi: \mathcal P \rightarrow \Xi $ is identifiable with
respect to $\mathcal G$ in the sense of Definition \ref{identDef}. 
\end{definition}

\section{Proof of proposition \ref{prop:LR}}
\label{appendix: proof of proposition 1}

	\begin{proof}

  To see that (\ref{claim1}) $\implies$ (\ref{claim2}),
  we note that the local characteristics for the baseline variables under $P$ coincide with the respective local characteristics under $\tilde P$. This means that $\tilde P|_{\mathcal F_0} = 
	P|_{\mathcal F_0}$. 
  Moreover, 
  note that the causal validity means that 
  $M^i_t  := N^i_t - \int_0^t \lambda_s^i ds$ is a local  $\mathcal F$-martingale with respect to $\tilde P$ for every $i \neq *$. On the other hand,  
  we have already assumed that $\tilde M_t^* := N_t^*  - \int_0^t \rho_s \lambda_s^* ds$ is  
  a local  $\mathcal F$-martingale with respect to $\tilde P$. Since we assume that    $\tilde P \ll P$, \cite[Theorem III.5.43]{JacodShiryaev} tells us that the likelihood-ratio process coincides with $W_t$. This means that (\ref{claim2}) also holds.

  Conversely, assume that (\ref{claim2}) holds. Note that this immediately  implies   $\tilde P  \ll P$ on $\calf_T$.   Moreover,  let $K^*_t := \int_0^t   ( \rho_s - 1 ) dM_s^*.$ We know \cite[Theorem II.37]{Protter} that $W$ is the (unique) solution of the  integral equation $W_t = 1 + \int_0^t W_{s-}dK_s^*.$
	   Now, whenever $h$ is a bounded and $\mathcal F$-predictable process,  we let $H =  \int_0^T h_s dM_s^i$, and note that (\ref{claim2}) implies that 
	  \begin{align*}
	&  E_{\tilde P}\bigg( \int_0^T h_s dM_s^i\bigg) = E_{ P}\bigg(W_T \int_0^T h_s dM_s^i\bigg) \\  &=   E_{ P}\bigg(   \int_0^Th_s  W_{s-}  dM_s^i + 
	 \int_0^t \int_0^{s-} h_r dM_r^i W_{s-} dK_s^ *  + \int_0^T   h_s W_{s-} d[K^*, M ^i]_s      
	 \bigg) \\ &=  0, 
	  \end{align*}   
	since $M^i$ and $M^*$ (and hence $K^*$) are local martingales with respect to $P$, and $K^*$ and $M^i $ are orthogonal.

	On the other hand, we have that 
  \begin{align*}
&  E_{\tilde P}\bigg( \int_0^T h_s d\tilde M_s^*\bigg) = E_{ P}\bigg(W_T \int_0^T h_s d\tilde M_s^*\bigg) \\  =  & E_{ P}\bigg(  \int_0^T h_s W_{s-} d\tilde M_s^* + 
 \int_0^t \int_0^{s-} h_r d\tilde M_r^* W_{s-} dK_s^*  + \int_0^T   h_s W_{s-} d[K^*, \tilde M ^*]_s      
 \bigg) \\ = &  
 E_{ P}\bigg\{  \int_0^T h_s W_{s-}  d\tilde M_s^* 
    + \int_0^T   h_s W_{s-} ( \rho_s - 1)N_s^*      
  \bigg\} \\ = &
   E_{ P}\bigg(   \int_0^T h_s W_{s-}  \rho_s d M_s^*\bigg)   = 0,  
  \end{align*}
	   since $M^*$ is a local martingale with respect to  $P$. This means that $\tilde M^*$ and  every $M^i$ for $i \neq *$ form local martingales with respect to $\tilde  P$.  
	   This is equivalent to causal validity, which means that (\ref{claim1})  holds. 
	\end{proof}

\section{Proof of Theorem \ref{theorem:omitU}}
 \label{appendix: proof of theorem omitU}

\begin{proof}[of theorem \ref{theorem:omitU}]
Suppose we have eliminability with a sequence $U^1,  \dots,  U^K$. 
By assumption,  $\delta$-separation implies  that either 
\begin{equation} \label{eq:u}
    \calu_1 
 \nrightarrow (\caln_0^{\backslash * } \cup \bar \calu^{2})
  \,|\, (\calv_0 \cup \bar \calu^{2})
\end{equation}
or 
\begin{equation} \label{eq:uu} 
    \calu_1  
  \nrightarrow N^*
  \,|\, ( \calv_0 \cup \bar \calu^{2}). 
\end{equation}

%The $\delta$-separations  imply that either 
%$\calu \nrightarrow \calv_0 \mid N^*$ or 
%$\calu \nrightarrow N^*\mid  \calv_0$.\\

First, assume that \eqref{eq:u} holds.  
We need to show that whenever $N$ is a counting process in $\calv_0^{\setminus *}\cup {\bar \calu}^2$, then its $\mathcal F_t^{\calv_0,  {\bar \calu}^2}$ intensities with respect to $P$ and $\tilde P$ coincide. 
(Note that we prove the claim only for processes; the case for baseline variables is analogous.)
Let  $N \in \calv_0^{\setminus *}\cup {\bar \calu}^2$,  and let $\lambda$ be  an $\mathcal F_t^{\calv_0,  {\bar \calu}^2}$-predictable intensity for $N$ with respect to $P$. The local independence property implies that  $\lambda$ also defines an intensity for $N$ with respect to 
$\mathcal F_t^{\calv_0,   \calu^1 ,  {\bar \calu}^2}$. However, since the model is causally valid over $\calv$, %this filtration,
$\lambda$ defines an intensity w.r.t. $\tilde P$ and $\mathcal F_t^{\calv_0,     \calu^1 ,  {\bar \calu}^2}$ as well. Finally, by assumption,  $\lambda$ is  $\mathcal F_t^{\calv_0 ,  {\bar \calu}^2}$-predictable, so it defines an intensity with respect to $\tilde P$ and $\mathcal F_t^{\calv_0 ,  {\bar \calu}^2}$, and the claim holds.

Assume instead that \eqref{eq:uu} holds. Since the model is causally valid with respect to  $\mathcal F_t^{\calv_0,    \calu^1 ,  {\bar \calu}^2}$, we have by  Proposition \ref{prop:LR} that 
the ratio 
$d\tilde P|_{\mathcal F_t^{\calv_0 ,     \calu^1 ,  {\bar \calu}^2} }/dP|_{\mathcal F_t^{\calv_0 ,     \calu^1 ,  {\bar \calu}^2}}$ 
%$\frac{d\tilde P|_{\mathcal F_t^{\calv_0 ,     \calu^1 ,  {\bar \calu}^2} }}{dP|_{\mathcal F_t^{\calv_0 ,     \calu^1 ,  {\bar \calu}^2}}}$ 
is of the form 
$$
W_t := \prod_{s \leq t} \rho_s^{\Delta N_s^*} \exp{ \Big\{  - \int_0^t ( \rho_s - 1) \lambda_s^* ds  \Big\}}
$$
 where $\lambda^*$ and $\rho \lambda^*$ coincide with  
 $\mathcal F_t^{\calv_0,     \calu^1 ,  {\bar \calu}^2}$ intensities for $N^*$ with respect $P$ and $\tilde P$. The local independence property now implies that 
 $\lambda^*$ and $\rho \lambda^*$ can be also  chosen to coincide with  $\mathcal F_t ^{\calv_0 , {\bar \calu}^2}$-predictable processes. This means that there exists an $\mathcal F_t ^{\calv_0 , {\bar \calu}^2}$-measurable version of $W_t$. One more application of Proposition \ref{prop:LR} then shows that the model is causally valid also with respect to the reduced filtration  
 $\mathcal F_t^{\calv_0  ,  {\bar \calu}^2}$.  The claim of Theorem \ref{theorem:omitU} now follows by simple induction.

\end{proof}

\section{Independent censoring and identifiability}
\label{appendix: independent censoring and identifiability}

We begin by  showing for a general multivariate counting process that under independent censoring with respect to a subset $\calv_0 \subset \calv \setminus \{ N^c \}$ of the whole system the $\calf^{\calv_0}$-intensities are identified based on the limited information  (i) ignoring processes not in $\calv_0$ and (ii) `stopping' at censoring.\\

Right censoring implies that for all $N\in \caln_0$ we can only ever observe the stopped processes $N_{t\wedge C}= \int_0^t I( s \leq C) dN_s$.
Note that the $\mathcal F^{\calv_0}$-intensity of a counting process $N\in \mathcal V_0$ with respect to $P$ defines a target parameter in the sense of \eqref{eq:xiDef}. Under censoring, $\mathcal F^{\calv_0}$ cannot be observed, as only the information stopped at censoring can be seen. We are thus interested in whether these $\mathcal F^{\calv_0}$-intensities are identified based on the observable information.   
This reduced, observable information is formally denoted by the $\sigma$-algebra  
$ \calf^{\calv_0 \cup N^c}_{T\wedge C}$ 
containing all potentially observable events for $\mathcal V_0$ and stopped at censoring.  Below, Lemma  \ref{lem:newDecens} says that independent  censoring implies that the $\mathcal F^{\calv_0}$-intensity of every $N \in \mathcal V_0$ is identifiable from the observable information, even if $\mathcal F^{\calv_0}$ contains unobservable events due to censoring.

 \begin{lemma}\label{lem:newDecens}
 Let $\mathcal P$ be a local independence model where we have independent
censoring in the model induced by $\mathcal V_0 \cup N^c$. Whenever $N \in \mathcal V_0$, there exists a non-negative and $\mathcal F_t^{\mathcal V_0}$-predictable process $\mu$ such that 
\begin{align}
    I(\cdot  \leq C) \mu  \label{eq: intensity factorisation}
\end{align}
 is an $\mathcal F^{\mathcal V_0 \cup N^c}_{t \wedge C}$-intensity for the stopped process $N_{\cdot \wedge C}$.

If moreover $P( t\leq  C | \mathcal F_t^{\mathcal V_0})   > 0$ $P$-a.s.\ whenever $t < T$, then $ \mu$ is an $\mathcal F_t^{\mathcal V_0}$-intensity w.r.t. $P$ for $N$.
 \end{lemma}

\begin{proof}
    Let $\lambda$ denote the $\mathcal F^{\mathcal V_0}$-intensity w.r.t. $P$ for $N$. Independent censoring means that $M_t := N_t - \int_0^t \lambda_s ds$ is a $P$-martingale with respect to the filtration $\mathcal F^{\mathcal V_0 \cup N^c}$. The optional stopping theorem now implies that 
    $I(\cdot  \leq C) \lambda $ is an $\mathcal F^{\mathcal V_0 \cup N^c}_{t \wedge C}$-intensity for the stopped process $N_{\cdot \wedge C}$, which proves the first claim. 
    
    Suppose $\mu$ is any process such that 
    $I(\cdot  \leq C) \mu $ is an $\mathcal F^{\mathcal V_0 \cup N^c}_{\cdot \wedge C}$-intensity for the stopped process $N_{\cdot \wedge C}$ and let $h$ be a bounded and $\mathcal F^{\mathcal V_0}$-predictable process.  Moreover, let  $\lambda$ be an    
     $\mathcal F^{\mathcal V_0}$-intensity for $N$ with respect to $P$. 
We have that 
\begin{align*}
    & E_P\Big\{\int_0^t h_s P ( s \leq C|\mathcal F_{s-}^{\calv_0}) ( \mu_s - \lambda_s) ds      \Big\} =   
    \int_0^t   E_P\Big\{ h_s P ( s \leq C|\mathcal F_{s-}^{\calv_0}) ( \mu_s - \lambda_s)\Big\} ds \\ = &    \int_0^t   E_P\Big\{ h_s 
    I  (s \leq C)  ( \mu_s - \lambda_s)\Big\} ds  =   
   E_P\Big\{ \int_0^t    h_s 
    I  (s \leq C)  ( \mu_s - \lambda_s) ds \Big\}\\ = &  
    E_P\Big( \int_0^{ t \wedge C}    h_s 
    dN_s \Big) -  E_P\Big( \int_0^{ t \wedge C}    h_s 
     dN_s \Big) =  0.    
\end{align*}
    Since this holds for any such $h$, and 
    $ P ( s \leq C|\mathcal F_{s-}^{\calv_0}) >0$ a.e., we have that $\mu_t = \lambda_t $ $P$-a.s. for almost every $t$.  Especially, we have that 
    $E_P \Big(\int_0^T h_s dN_s\Big) =  
    E_P \Big(\int_0^T h_s \mu_s ds \Big)$ for every $\mathcal F_t^{\mathcal V_0}$-predictable and bounded $h$, which proves the second claim.  
    
\end{proof}

\begin{remark}\label{remark:decens}
Lemma \ref{lem:newDecens} says that independent censoring allows us to identify the intensities for the uncensored processes. This means that, as long as the following positivity $P( t\leq  C | \mathcal F_t^{\mathcal V_0})   > 0$ holds $P$-a.s whenever $t < T$, there exists  \emph{de-censoring} maps $\zeta$ that construct $P$ restricted to $\calf_{t}^{\calv_0}$ from $P$ restricted to $\calf_{t\wedge C}^{\calv_0 \cup N^c}$, since $P$ is uniquely characterised by these intensities \cite[Theorem III 1.26]{JacodShiryaev}. Thus, as long as the above positivity assumption holds,
\begin{align}
    \zeta(P|_{\calf_{T\wedge C}^{\calv_0 \cup N^c}}) =  P|_{\calf_{T}^{\calv_0}}. \label{eq: zeta map}
\end{align}
If we have a $\mu$ as in \eqref{eq: intensity factorisation}, it can be used for an explicit construction of the map $\zeta$ in \eqref{eq: zeta map} as multiplication of a likelihood; see \citep{Jacod:Multivariate}.
\end{remark}

\section{Example of faithfulness violation}
\label{appendix: example of faithfulness violation}

Consider counting processes $N^A_t := I(A \leq t)$ and death $N^D_t := I(D \leq t)$ with intensities
\begin{align*}
    \lambda^A_t &= Y_t^A(U  + 1) \\
    \lambda^D_t &= Y_t^D(U + g(t,A)(1 - Y_t^A) + h_t)
\end{align*}
with respect to $\F_t^{A,D,U}$ and $P$. Here, $Y_t^D = I(N_{t-}^D=0)$, and $Y_t^A = I(N_{t-}^D=N_{t-}^A = 0)$, and $U$ is a binary variable such that $\gamma = \frac{P(U=1)}{P(U=0)}$. This system is compatible with the local independence graph:
\begin{equation*} %\label{total}
        \xymatrix{ 
        U \ar@/^/[dr] \ar@/^1pc/[drr]& & \\
          & N^A  \ar@/_/[r] & N^D \ar@/_/[l]
        }
\end{equation*}
Conditional expectations of random variables with respect to the natural filtration of a multivariate counting process $N$ has an explicit representation in terms of the jump times of $N$ \cite{Bremaud, LastBrandt1995marked, jacobsen2006point} (see e.g. \cite[Corollary 2.3.2]{LastBrandt1995marked} for an explicit result). This gives rise to the representation
\begin{align}
    \begin{split}
        Y_t^D E[U | \mathcal F_{t-}^{A,D}] = Y_t^D & \Big \{ \frac{E[ U I(t < A \wedge D) ]}{P(t < A \wedge D)}I(t \leq A \wedge D)  \\
        &+ \frac{E[ U I(t < A \vee D) | \F_{A \wedge D}]}{P(t < A \vee D| \F_{A \wedge D})}I(A \wedge D < t \leq A \vee D) \Big \}.
    \end{split}
    \label{eq: U condExp}
\end{align}
We calculate the bracketed terms on the right hand side of \eqref{eq: U condExp}. For the first term we obtain by Bayes' rule
\begin{align*}
    \frac{E[ U I(t < A \wedge D) ]}{P(t < A \wedge D)} &= \frac{P(t < A \wedge D|U=1)P(U=1)}{P(t < A \wedge D|U=1)P(U=1) + P(t < A \wedge D|U=0)P(U=0)} \\
    &= \frac{\gamma}{\gamma + \frac{P(t < A \wedge D|U=0)}{P(t < A \wedge D|U=1)}} \\
    &= \frac{\gamma}{\gamma + e^{ 2t }}.
\end{align*}
For the second term on the right hand side of \eqref{eq: U condExp} we only need to treat the case $A<D$, because $\{A \geq D\}$ is a null set. Note that $\F_{A \wedge D} \cap \{ A < t \leq  D \} = \sigma(A, A<D) \cap \{ A < t \leq  D \}$. Furthermore,
\begin{align*}
    P(t < A \vee D | A, A<D, U) &= e^{ -\int_0^{A \wedge t} \lambda_s^A + \lambda_s^D ds - \int_A^{A \vee t} \lambda_s^D ds } \\
    &= e^{ -\int_0^{A \wedge t} 2U +1 + h_s ds  - \int_A^{A \vee t}  U + g(s,A) + h_s ds }
\end{align*}
This gives 
\begin{align*}
    &\frac{E[ U I(t < A \vee D) | \F_{A \wedge D}]}{P(t < A \wedge D| \F_{A \wedge D})}I(A < t \leq D) = \frac{E[ U I(t < A \vee D) | A , A < D]}{P(t < A \wedge D| A , A < D)}I(A < t \leq D) \\
    &= \frac{\eta^{A<D} }{\eta^{A<D} + \frac{P(t < A \vee D | A , A < D,U=0)}{P(t < A \vee D | A , A < D,U=1)} }I(A < t \leq D) \\
    &= \frac{\eta^{A<D} }{\eta^{A<D} + e^{ A + t } }I(A < t \leq D),
\end{align*}
 where we have defined
 $$\eta^{A<D} := P(U=1 | A, A<D) / P(U=0 | A, A<D).$$ 
In conclusion we have that \eqref{eq: U condExp} equals
\begin{align*}
    %\begin{split}
     Y_t^D E[U | \mathcal F_{t-}^{A,D}] &= Y_t^D \Big \{ \frac{\gamma}{\gamma + G_t } I(t \leq A \wedge D) + \frac{\eta^{A<D} }{\eta^{A<D} + G_t }I(A < t \leq D) \Big\},
    %\end{split}
    %\label{eq: in conclusion}
\end{align*}
where $G_t := e^{ t \wedge A + t }$.

Using the innovation theorem we find that 
\begin{align*}
    \bar \lambda^A_t &:= E[ \lambda^A_t | \mathcal F_{t-}^{A,D} ] = Y_t^A \big( E[U |\F_{t-}^{A,D}]  + 1 \big) \\
    \bar \lambda^D_t &:= E[ \lambda^D_t | \mathcal F_{t-}^{A,D} ] = Y_t^D\big( E[U |\F_{t-}^{A,D}] + g(t,A)(1 - Y_t^A) + h_t \big)
\end{align*}
define $\F^{A,D}_t$-intensities of $N^A$ and $N^D$, respectively. Pick $g(t, A) = - \frac{\eta^{A < D}}{\eta^{A < D} + G_t}$, and $h_t = - \frac{\gamma}{\gamma + G_t}Y_t^A + B$ for sufficiently large $B$ to ensure that $\lambda^D \geq 0$. We then obtain 
\begin{align} \label{eq:noncausal1}
    \bar \lambda^D_t &= Y_t^D B, \\
    \bar \lambda^A_t &= Y_t^A \Big( \frac{\gamma}{\gamma + G_t } + 1 \Big) \notag
\end{align}
which shows that the local independence graph for the system where $U$ is marginalised out is 
\begin{equation*} %\label{total}
        \xymatrix{ 
          & N^A & N^D \ar[l]
        } .  
\end{equation*}

\subsection*{Causal validity does not hold in the marginalised system}
Suppose that the model for $(U, N^A, N^D)$ is causally valid w.r.t. an intervention that would prevent $N^A$ from jumping, i.e. that $\tilde P (t > A ) = 0  $.  This means that the local characteristics  w.r.t. $\tilde P$ would be given by   $\tilde \lambda^A = 0$, $\tilde \lambda^D = Y^D ( U + h)$, and $\tilde P(U = 1) = P( U = 1) $, where 
$
h_t = B - \frac{ \gamma } { \gamma + e^{2t}} 
$
$\tilde P ~a.s$.

We furthermore have by Bayes' rule that 
\begin{align*}
    \tilde P( U = 1 | D \geq   t ) = \frac{ \gamma  }   {\gamma + \frac { \tilde P( D \geq t | U = 0)   } {\tilde P( D \geq t | U = 1) }   } = 
    \frac{ \gamma  }   {\gamma + e^{ t}  }. 
\end{align*}
By the innovation theorem, we have that the intensity of $N^D$ w.r.t. $\{\F^{A,D}_t\}_t$ and $\tilde P$ coincides with 
\begin{align} \label{eq:noncausal2}
E_{\tilde P}[\tilde \lambda_t^D |\mathcal F_{t-}^{A,D} ] = & Y^D_t 
( E_{\tilde P }[U  |\mathcal F_{t-}^{A,D} ]  + h_t) 
= Y_t^D  ( \frac{ \gamma  }   {\gamma + e^{ t}  } + h_t)  \\ 
= & Y_t^D  (  B + \frac{ \gamma  }   {\gamma + e^{ t}  }  -   \frac{ \gamma  }   {\gamma + e^{2 t}  }). 
\end{align}
If the submodel induced by $N^A$ and $N^D$ had been causally valid w.r.t. this intervention, then the right-hand side of \eqref{eq:noncausal2} would coincide with \eqref{eq:noncausal1}, which is clearly not the case.

\section{Proof of Proposition \ref{proposition:indCens}}
\label{appendix: proof of proposition indCens}

\begin{proof}[of Proposition \ref{proposition:indCens}] 

Recall that  interventions on counting processes do not affect baseline variables. We start by showing that $P$ and $\tilde P$ are the same when restricting to $\calf^{\calv \setminus N^c}$. To that end we consider a counting process $N \in  \caln \setminus \{N^c\}$ with $\calf$-intensity $\lambda$ with respect to $P$. From causal validity $\lambda$ is also the $\calf$-intensity of $N$ with respect to $\tilde P$. Independent censoring means that ch$(N^c) = \emptyset$, and in particular that $\lambda$ is the $\calf^{\calv \setminus N^c}$-intensity of $N$ with respect to $P$. By analogous reasoning as in the proof of Theorem \ref{theorem:omitU},
%appealing to Lemma \ref{lemma:moving down} two times (i.e.\ for both $P$ and $\tilde P$) 
we obtain that $\lambda$ is also the $\calf^{\calv \setminus N^c}$-intensity of $N$ with respect to both $P$ and $\tilde P$. %which again yields that $\lambda_t$ is the $\mathcal F_t$-intensity of $N_t$ with respect to $\tilde P$ due to causal validity. From  
%as $\lambda_t$ is $\calf_t^{\calv \setminus \{ N^c \}}$-predictable, we get that $N_t - \int_0^t \lambda_s ds$ defines an $\calf_t^{\calv \setminus \{ N^c \}}$-adapted martingale with respect to $\tilde P$.

We conclude that $P|_{\mathcal F_T^{\calv \setminus N^c} } = \tilde P|_{\mathcal F_T^{\calv \setminus N^c}}$. As ${\mathcal F}^{\calv_0}_T \subset\calf_T^{\calv \setminus N^c}$, this implies $P|_{{\mathcal F}^{\calv_0}_T } =\tilde P|_{{\mathcal F}^{\calv_0}_T}$.

Now, pick a  counting process $N \in \calv_0$ with $\calf^{\calv_0}$-intensity $\tilde \nu$ with respect to $\tilde P$, and let $h$ be an $\calf^{\calv_0}$-predictable process. Then 
 \begin{align*} %\label{eq:causalCond}
                E_{ P} \Big( \int_0^T   h_s  dN_s \Big) 
                =E_{\tilde P} \Big( \int_0^T   h_s  dN_s \Big) 
                =E_{\tilde P} \Big( \int_0^T   h_s \tilde \nu_s ds \Big) 
			= E_{P} \Big( \int_0^T   h_s \tilde  \nu_s ds \Big)
                \end{align*}
where the first and third equalities are due to $P|_{{\mathcal F}^{\calv_0}_T } =\tilde P|_{{\mathcal F}^{\calv_0}_T}$ and the second by definition. Hence, $\tilde \nu$ is also the $\calf^{\calv_0}$-intensity of $N$ with respect to $P$.
         
Finally, the assumptions of the proposition imply that $\caln_0$ is locally independent of $ N^c$ given $\calv_0$, so the $\calf^{\calv_0}$-intensity of any $N \in \calv_0$ is identical with its $\calf^{\calv_0 \cup N^c}$-intensity under $P$ by definition. The analogue is true under $\tilde P$. This implies that 1) the model restricted to $\calf^{\calv_0 \cup N^c }$ is subject to independent censoring, and 2) the $\calf^{\calv_0 \cup N^c}$-intensity of each $N \in \calv_0$ is the same under $P$ and $\tilde P$.
 \end{proof}

%\begin{corollary} \label{corollary: prevent censoring}
%Let $\calp(G)$ be as in Proposition \ref{proposition:indCens} and assume that the model is causally valid with respect to prevention of censoring, with $P^p$ being the model in which censoring is prevented. Let $P^r$ be the model where we have imposed an $\calf^{\calv_0\cup N^c}$-predictable intensity for censoring. Then, we have that $$ P^p|_{\calf_{t}^{\calv_0}} =  P^r|_{\calf_{t}^{\calv_0}} = P|_{\calf_{t}^{\calv_0}}. $$
%In particular, every parameter on such hypothetical measures restricted to $\calf^{\calv_0}$ is invariant  with respect to choice of causally valid censoring strategy as in Proposition %\ref{proposition:indCens}.
%%No censoring corresponds to a censoring intensity identically equal to 0. As this intensity is predictable, the statement of Proposition \ref{proposition:indCens} also holds for prevention of censoring (as long as the model is causally valid with respect to this intervention). In particular, 
%\end{corollary}

\begin{proof}[of Corollary \ref{corollary: prevent censoring}] 
%Proof of Corollary \ref{corollary: prevent censoring} 
%\begin{proof}
    As the model is causally valid with respect to preventing censoring, it is also causally valid with respect to randomising censoring (see Remark \ref{rem:rondomcens}), and Proposition \ref{proposition:indCens} ensures causal validity restricted to $\calf^{\calv_0 \cup N^c}$ for both. By definition of causal validity, the $\calf^{\calv_0 \cup N^c}-$intensity $\lambda$ of every $N \in \calv_0$ is the same after each given intervention. This means that $\lambda$ is the $\calf^{\calv_0 \cup N^c}-$intensity with respect to both $P^p$ and $P^r$. As the model restricted to $\calf^{\calv_0 \cup N^c}$ is subject to independent censoring (again by Proposition \ref{proposition:indCens}), $\lambda$ is also the $\calf^{\calv_0}-$intensity with respect to both $P^p$ and $P^r$. Uniqueness of the intensities gives the desired result.
\end{proof}

\section{Proof of Theorem \ref{theo:1}}
\label{appendix: proof of theorem 1}

We are now ready to prove Theorem \ref{theo:1}, which we restate.

Remember that the interventions replace the $\calf^{\calv}$-intensities of this system by new intensities for $N^c$ and $N^x$. As we cannot observe $\calu$ we work with the observable intensities, where $ \lambda^c$ is the $\mathcal F^{\mathcal V_0 \cup \mathcal L \cup N^c}$-intensity of $N^c$ with respect to $P$, while $ \lambda^x$ denotes the $\mathcal F^{\mathcal V_0 \cup\mathcal L \cup N^c}$-intensity of $N^x$ with respect to $P$. In contrast, the interventions enforce an $\mathcal F^{\mathcal V_0 \cup N^c}$-predictable $\tilde \lambda^c := \rho^c \cdot \lambda^c$, and an $\mathcal F^{\mathcal V_0}$-predictable $\tilde \lambda^x := \rho^x \cdot \lambda^x$.

\setcounter{theorem}{1}  
\begin{theorem} %\label{theo:1}
                With the notation and set-up as in Section \ref{sec:msm}, consider a local independence model $\calp(G)$. We assume that censoring is independent with respect to $\calv$, i.e.\ ch$(N^c)=\emptyset$ in $G$.
 We assume causal validity with respect to an intervention that prevents censoring,  and  for additionally intervening on the treatment process. Let $\tilde {\mathcal P}$ denote the resulting interventional model.
                        Further, assume                    
                        \begin{enumerate}[(i)]
  \item %\label{item: thm 1 delta separation}
  $N^c \nrightarrow_G (\caln_0, \caln_{\call}) | (\calv_0,\call)$, and  %\\[1mm]
  \item %\label{item: thm 1 eliminable}
  $\calu$ is eliminable with respect to $(N^x, \mathcal L \cup \mathcal V_{0}^{\backslash x})$;
         \end{enumerate}
finally, assume the technical conditions that         %$\frac{\tilde{\lambda}^c } { {\lambda}^c }$
$\rho^c$ and %$\frac{\tilde{\lambda}^x } { {\lambda}^x} $
$\rho^x $ are bounded. Then we  have:\\ [1mm]
The interventional distribution $\tilde{P}$ under both hypothetical  interventions (preventing censoring and intervening on treatment) restricted to the subset $ \mathcal V_0$ is identified  from the  observable information $ {\mathcal F_{T \wedge C}^{\mathcal V_0 \cup \mathcal L \cup N^c}}$ by
\begin{align*}
\tilde P| _{  \mathcal F_{t}^{\mathcal V_0}} 
&= 
                          \zeta\bigg\{ \bigg( W_{t\wedge C} \cdot P|_{\mathcal F_{t\wedge C}^{\mathcal
                                        V_0 \cup \mathcal L \cup N^c}}   
                        \bigg)\bigg|_{\mathcal F_{t\wedge C}^{\mathcal V_0 \cup
                                        N^c}}
                        \bigg\}.
\end{align*}
 Thus, the marginal density $\tilde P|_{ \mathcal F^{\mathcal V_0 }_t }$ is given by re-weighting, marginalising over $\call$ and applying a de-censoring map as in \eqref{eq: zeta map txt} (see Remark \ref{remark:decens} in Appendix \ref{appendix: independent censoring and identifiability}).
 
\end{theorem}

\begin{proof}
	
	We want to identify the hypothetical distribution $\tilde{P}$ under both interventions (preventing censoring and intervening in treatment) restricted to the processes in $ \mathcal V_0$. We will do this by first obtaining the likelihood-ratio between the observational and hypothetical distribution restricted to the filtration $\calf^{\calv_0 \cup \call \cup N^c}_{t \wedge C}$. 
	%We will then use Lemma \ref{lem:newDecens} to  obtain $\tilde P|_{\calf_t^{\calv_0\cup \call}}$. Finally, $\tilde P|_{\calf_t^{\calv_0}}$ is obtained by marginalising over $\call$. 
Here, we need to show that the weights $W_{T\wedge C}$ are identified given the observable information $ {\mathcal F_{T \wedge C}^{\mathcal V_0 \cup \mathcal L \cup N^c}}$.  This follows from the first two steps:
We argue that $\mathcal U$ can be `ignored' without destroying causal validity regarding censoring and treatment; this is based on  Theorem \ref{theorem:omitU}  and Proposition \ref{proposition:indCens}  (steps 1 and 2, below). %It will be important that this causal validity holds with respect to a number of interventional scenarios:                 
We let $\mathcal P^p$ correspond to the model for an intervention where censoring is prevented (no intervention on the treatment process), %$\mathcal P^2$
$\mathcal P^r$ corresponds to the model for an intervention where censoring  is randomised according to  $\tilde \lambda^c$ (no intervention on the treatment process). %These are then both modified to include an intervention on the treatment process, denoted by $\tilde{\mathcal P}$ and $\mathcal P^r$, respectively. Here, 
%$\tilde P$ is the measure in which we are ultimately interested under the further intervention on the treatment process. %, where $P'$ is only needed to define $\tilde P$.
%However, the weights $W_t$ from \eqref{cx.weights} are in fact those with which we obtain $P^r$. 
%Hence, we argue that we can `move' from $P$ to $P^r$ by re-weighting, and that this is identified from the observable processes, so marginal over $\mathcal U$ and stopped at censoring  (Step 3) --- see equation \eqref{eq:MSMLR} below.
%We use Proposition \ref{proposition:indCens} and \ref{theorem:omitU} to establish causal validity with respect to interventions on both censoring and treatment when restricted to $\calf^{\calv_0 \cup \call}$ (step 1 and 2), 
Then  we show that the likelihood-ratio between $\tilde P$ and $P$ is given by \eqref{cx.weights} in (step 3, below). Finally, we argue with Lemma \ref{lem:newDecens} that we can extend the measure $\tilde P$ beyond the censoring time (Step 4).\\
%Finally, we argue that we can `move'  from $P^r$ to $\tilde P$ --- see equation \eqref{eq:MSMinvar} below (Step 5).  These steps are then combined to obtain the overall result \eqref{eq:paramsms}.\\

{\em Step 1 (marginalising $\calu$, retaining causal validity wrt.\ preventing or randomising censoring):}

By Proposition \ref{proposition:indCens}, condition \eqref{item: thm 1 delta separation} ensures causal validity with respect to prevention of censoring (and hence randomisation of censoring by assumption) when restricted to $\calv_0 \cup \call \cup \{ N^c\}$; thus $\calu$ can be ignored regarding censoring. 

Moreover, with Corollary \ref{corollary: prevent censoring}, we conclude that $P^r$ coincides with $P^p$ when restricted to $\calf^{\calv_0 \cup \call}$, and hence when restricted to $\calf^{\calv_0}$ as $\calf^{\calv_0} \subset \calf^{\calv_0 \cup \call }$. %As we are interested in identification restricted to $\calf^{\calv_0}$, we conclude that we are free to choose censoring strategy (preventing or randomising). 
We can thus consider the strategy that randomises censoring, i.e.\ $P^r$, in the remainder of the proof. %$P^r$ and $P^p$ coincide when restricted to $\calf_t^{\calv_0}$  as $\calf_t^{\calv_0} \subset \calf_t^{\calv_0 \cup \call}$. 
\\

{\em Step 2 (marginalising $\calu$, retaining causal validity wrt.\ treatment):}
                As censoring is independent in $\calp(G)$ and with  \eqref{item: thm 1 delta separation}  and \eqref{item: thm 1 eliminable}, we have that $\calu$ is eliminable with respect to $(N^x,\call  \cup \{ N^c\}\cup \calv_0^{\backslash x} )$, which also holds under $P^r$ as no new dependencies are introduced. 
                From Theorem \ref{theorem:omitU} we thus have that the restriction of $P^r$ to $\mathcal V_0 \cup \mathcal L \cup \{ N^c\}$ (ignoring $\calu$) remains causally valid with respect to the treatment intervention. \\ %\textcolor{red}{\textbf{Pål:} I don't see how the argument in step 2 holds. Theorem \ref{theorem:omitU} gives us that the $\mathcal F^{\mathcal V_0 \cup \mathcal L \cup \{ N^c\}}$-intensities \textit{of the processes in} $\caln_0^{\setminus x}$ are the same under $P^r$ and $\tilde P$; we don't know what the intensities of the processes in $\mathcal L$ are (whether they are the same under $P^r$ and $\tilde P$). So we can not conclude that causal validity holds. %But one would still expect such a result to hold..                 What we \textit{do} know is that the likelihood ratio $R_t^r := \frac{d\tilde P}{dP^r} |_{\mathcal F_t^{\calv_0 \cup \call \cup \{ N^c \}}}$ takes a certain form given by products over nodes in $\{ N^x \} \cup \call \cup \{ N^c \}$:                 $$ R^r = S^x \times S^c \times \prod_{z \in \call} S^z, $$                 where we know what $S^x$ is (but not the other terms). We would need to 'integrate out' the other terms.. } \\

{\em Step 3 (combined intervention by using combined weights):}

%!!! make sure to emphasise the ordering of the interventions!!\\

In Step 1 we argued that causal validity for the intervention that randomises censoring is retained when ignoring $\mathcal U$. The likelihood-ratio associated with this intervention is 
     \begin{equation*}
       \frac{ dP^r |_{\mathcal F_t^{\mathcal V_0 \cup \mathcal L \cup  N^c }}} {
               dP|_{\mathcal F_t^{\mathcal V_0 \cup \mathcal L \cup   N^c }} }=
              (\rho_C^c)^{I(t
               \geq C)} \exp \left\{-\int_0^t (\rho_s^c  - 1 ) \lambda_s^c ds  \right\}.
     \end{equation*}
              
Moreover, in Step 2 we argued that causal validity for the intervention on the treatment process is retained when ignoring $\mathcal U$.
This imposes the intensity $\tilde \lambda^x$ for the counting process $N^x$, and the likelihood-ratio associated with this intervention is \\
  \begin{equation*}
     \frac{ d\tilde P|_{\mathcal F_t^{\mathcal V_0 \cup \mathcal L \cup  N^c }}} {
             dP^r |_{\mathcal F_t^{\mathcal V_0 \cup \mathcal L \cup 
                             N^c }} }=
     \prod_{s \leq t} %\left( \frac{\tilde \lambda_s^x} {
                     %{\lambda}_s^x} \right)
                     (\rho_s^x)^{\Delta N_s^x} \exp \left\{-\int_0^t (\rho_s^x
             %\tilde \lambda_s^x  - 
             -1){\lambda}_s^x ds  \right\}. 
        \end{equation*}
               
Together, this implies that re-weighting  $P$ leads to $\tilde P$, ignoring $\mathcal U$, using  the following weights:
   \begin{equation*}
      \frac{ d\tilde P|_{\mathcal F_t^{\mathcal V_0 \cup \mathcal L \cup                               N^c }}} {
              dP|_{\mathcal F_t^{\mathcal V_0 \cup \mathcal L \cup 
                              N^c} } } =
      \frac{ d\tilde P|_{\mathcal F_t^{\mathcal V_0 \cup \mathcal L \cup                               N^c }}} {
              dP^r |_{\mathcal F_t^{\mathcal V_0 \cup \mathcal L \cup                               N^c }} }
      \cdot  \frac{ dP^r |_{\mathcal F_t^{\mathcal V_0 \cup \mathcal L \cup                               N^c }}} {
              dP|_{\mathcal F_t^{\mathcal V_0 \cup \mathcal L \cup
                              N^c }} } = W_t, 
   \end{equation*}
where  $W_t$ is defined as in (\ref{cx.weights}).
The process $W_t$ is not observable after censoring. However, by the optional stopping theorem we have that 
                
                %???or can we work with the non-stopped filtration all the time after modifying Prop 3??? Kjetil to think about this!!!
                
\begin{equation} \label{eq:MSMLR}
         W_{t\wedge C} \cdot P|_{\mathcal F_{t\wedge C}^{\mathcal
                                        V_0 \cup \mathcal L \cup N^c}} 
                        = \tilde P|_{\mathcal F_{t\wedge C}^{\mathcal
                                        V_0 \cup \mathcal L \cup N^c}}. 
\end{equation}
Hence the stopped weights are identified by the observed information; as before we rely on the uniqueness of the intensities \cite[Theorem II.T12]{Bremaud}.\\

\smallskip

{\em Step 4:} 
Now, from \eqref{eq:MSMLR} we have $\tilde P$ expressed in terms of the weights and the observational density stopped at censoring, restricted to $\calf_{t\wedge C}^{\calv_0 \cup \call \cup N^c}$. This means that 
\begin{align*}
                        W_{t\wedge C} \cdot P|_{\mathcal F_{t\wedge C}^{\mathcal
                                        V_0 \cup \mathcal L \cup N^c}}   
                       \bigg|_{\mathcal F_{t\wedge C}^{\mathcal V_0 \cup
                                        N^c}}
                        &=  \tilde P|_{\mathcal F_{t\wedge C}^{\mathcal
                                        V_0 \cup \mathcal L \cup N^c}}
                        \bigg|_{\mathcal F_{t\wedge C}^{\mathcal V_0 \cup
                                        N^c}} %_\theta 
                        =   \tilde P|_{\mathcal F_{t\wedge C}^{\mathcal V_0 \cup N^c}}.
\end{align*}
Moreover, we have that independent censoring holds under $\tilde P$ within  $\calv_0\cup N^c$. This is because  $\call \nrightarrow N^c | \calv_0$ under $\tilde P$  due to construction (as $\tilde \lambda^c$ is $\calf^{\calv_0}$-predictable). Thus, with the same argument as for eliminability, property \eqref{item: thm 1 delta separation} is retained when marginalising over $\call$ so that $N^c \nrightarrow \caln_0 | \calv_0$. With Proposition \ref{proposition:indCens}, we see that independent censoring holds
 in the model restricted to $\calv_0 \cup N^c$ under $\tilde P$. The de-censoring map  $\zeta$ of \eqref{eq: zeta map} thus yields the distribution with respect to the non-stopped filtration:
\begin{align*}
                          \zeta\bigg\{ \bigg( W_{t\wedge C} \cdot P|_{\mathcal F_{t\wedge C}^{\mathcal
                                        V_0 \cup \mathcal L \cup N^c}}   
                        \bigg)\bigg|_{\mathcal F_{t\wedge C}^{\mathcal V_0 \cup
                                        N^c}}
                        \bigg\}%_\theta
                        &= 
                        \zeta\bigg\{\bigg(\tilde P|_{\mathcal F_{t\wedge C}^{\mathcal
                                        V_0 \cup \mathcal L \cup N^c}}
                        \bigg)\bigg|_{\mathcal F_{t\wedge C}^{\mathcal V_0 \cup
                                        N^c}} \bigg\}%_\theta 
                        =   \zeta\bigg(\tilde P|_{\mathcal F_{t\wedge C}^{\mathcal V_0 \cup N^c}}\bigg)%_\theta
 =  \tilde P| _{  \mathcal F_{t}^{\mathcal V_0}}.  
 %= 
% \bigg(\tilde P| _{  \mathcal F_{t}^{\mathcal V_0 \cup \mathcal L}}  
% \bigg) \bigg| _{  \mathcal F_{t}^{\mathcal V_0}}  
%=
% \bigg(\tilde P| _{  \mathcal F_{t}^{\mathcal V_0 \cup \mathcal L}}  
% \bigg) \bigg| _{  \mathcal F_{t}^{\mathcal V_0}}  
%= \tilde P|_{\mathcal F_t ^{\mathcal V_0}}, 
\end{align*}
                %where $\zeta$ was defined in Remark \ref{remark:decens}.
% where the first equality is due to \eqref{eq:MSMLR}, the second one is due to the law of nested expectations, and the last is due to Lemma \ref{lem:newDecens}. %the fifth is due to  \eqref{eq:MSMinvar}, 
 
Finally, note that the technical conditions are required so that the process (\ref{cx.weights}) is uniformly integrable in order to define a proper likelihood-ratio; weaker conditions can also be given \cite[Theorem IV 4.16a]{JacodShiryaev}.
        \end{proof}

\section{Reproducible research}
 We do not have the rights to share the HPV dataset. The data can be obtained upon application at the Cancer Registry of Norway. The \texttt{R} implementation of the analysis in Section \ref{section: analysis}, which has been applied to a simulated dataset, is available on the GitHub repository \texttt{github.com/palryalen/paper-code}. The simulated dataset, \texttt{sim\_data.RData}, has identical features to the real dataset. These features include event types "CIN2+" and "follow-up", respectively indicating the occurrence of CIN2+ and subsequent testing for a given individual, along with event times.

A description of the data pre-processing used in the analysis in Section \ref{section: analysis} follows. The relevant features 
(test technology, event type, and event times) were extracted from the HPV dataset. Event time resolution was given with an accuracy of one month. We applied standard tiebreaking of the tied event times to use the \texttt{ahw} package. The data was organized in long format, identical to the format found in \texttt{sim\_data.RData}. The code for the  analysis of Section \ref{section: analysis} is displayed in the file \texttt{example\_analysis.R}. Both of these files can be found in the  GitHub repository.

\end{document}